\documentclass[12pt]{amsart}
\usepackage{epsf,amscd,amsmath,amssymb,amsfonts,mathtools}
\usepackage[colorlinks=true,linkcolor=red,citecolor=blue]{hyperref}
\usepackage{graphicx}
\usepackage{epstopdf}
\usepackage{fullpage}
\usepackage[utf8]{inputenc}
\usepackage{fullpage}
\usepackage{latexsym}
\usepackage{tikz}
\usepackage{color}

\newcommand{\setdiff}{\! \setminus \!}
\newcommand{\set}[1]{ [{#1}] }
\newcommand{\m}[1]{\mathcal{#1}}
\newcommand{\mb}[1]{\mathbb {#1}}

\theoremstyle{plain}
\newtheorem{thm}{Theorem}[section]
\newtheorem*{thm*}{\bf Theorem }
\newtheorem{lem}[thm]{Lemma}
\newtheorem{prop}[thm]{Proposition}
\newtheorem*{prop*}{\bf Proposition}
\newtheorem{cor}[thm]{Corollary}

\theoremstyle{definition}
\newtheorem{defn}[thm]{Definition}

\newtheorem{eg}[thm]{Example}

\theoremstyle{remark}
\newtheorem{rem}[thm]{Remark}

\numberwithin{equation}{section}

\newcommand{\be}{\begin{equation}}
\newcommand{\ee}{\end{equation}}
\newcommand{\bea}{\begin{eqnarray}}
\newcommand{\eea}{\end{eqnarray}}
\newcommand{\beas}{\begin{eqnarray*}}
\newcommand{\eeas}{\end{eqnarray*}}

\theoremstyle{plain}
\theoremstyle{definition}

\numberwithin{thm}{section}
\numberwithin{equation}{section}

\def\One{\mathbb{I}}

\usepackage{float}
\floatstyle{boxed} 
\restylefloat{figure}

\title{Graph Complexes and Feynman Rules}
\author{Marko Berghoff and Dirk Kreimer}
\address{Humboldt U.\ Berlin, Unter den Linden 6, 10099 Berlin, Germany}
\begin{document}
\maketitle
\begin{abstract}
We investigate Feynman graphs and their Feynman rules from the viewpoint of graph complexes.
We focus on the interplay between graph homology, Hopf-algebraic structures on Feynman graphs and the analytic structure of their associated integrals.  Furthermore, we discuss the appearance of cubical complexes where the differential is formed by reducing internal edges and by putting edge-propagators on the mass-shell.
\end{abstract}

\tableofcontents

Mathematics Subject Classification (MSC2020): 
81T15, 81Q30, 18G85, 57T05, 14D21.

\section{Introduction}

\subsection{Motivation}
Feynman integrals and graph complexes both belong arguably to the most mysterious objects populating modern mathematical physics. They have rather simple definitions, yet we only have a very limited understanding of the general structures underlying these objects. These structures appear to be very fundamental as both graph complexes and Feynman integrals are connected to many different areas of mathematics. For graph complexes these areas include the study of embedding spaces, the deformation theory of operads, the cohomology of various groups and Lie algebras, and the topology of moduli spaces, just to give a few examples. We refer to the original work of Kontsevich \cite{Kont1,Kont2} as well as \cite{Willwacher,ConantVogtmann,KarenV_MSRI,ChanGalatiusPayne} for further reading.\footnote{It is difficult to give a concise survey on graph complexes as they come in many variants; graphs may be decorated with additional data, satisfy certain relations/symmetries etc.} Feynman integrals on the other hand, apart from being the central objects in perturbative quantum field theory, are connected to the study of periods and special functions in number theory \cite{FBrown, FrancisI, FrancisII, PanzerHepp} as well as fundamental questions in modern algebraic geometry \cite{BEK,BlKrLMHS}. In addition, the discrete shadows of these integrals, Feynman graphs or diagrams, have a rich combinatorial structure which reaches into the fields of Hopf algebras \cite{core,BV,Kr-Y} (with plenty of applications from combinatorics to stochastic analysis) and even as far as category theory \cite{Kaufmann}.

In the present paper we aim at drawing a connection between the two fields, that is, we investigate the role graph complexes play in the study of Feynman integrals in perturbative quantum field theory.\footnote{For the opposite direction, see \cite{Francis-gc,Francis-hgc}.}
In the following we write $\Phi(G)$ for the Feynman integral associated to a Feynman graph $G$. Our goal is to study the analytic structure of $\Phi(G)$, viewed as a function of its kinematic variables, and clarify the role two particular graph complexes play in this endeavour,
\begin{itemize}
    \item a ``traditional" graph complex, generated by Feynman graphs $G$, the differential defined by a (signed) sum over all possible edge-collapses, \begin{equation*}                                                                                                                                      
        d:G \longmapsto \sum_{e \in E(G)} \pm G/e,                                                                                                                                           \end{equation*}
    \item a cubical chain complex whose generators are pairs $(G,F)$ where $F\subset G$ is a spanning forest and the differential is the (signed) sum of two maps, summing over all ways of collapsing or removing edges in $F$,
    \begin{equation*}
     d=d_0 + d_1:G \longmapsto \sum_{e\in E(F)} \pm \Big( (G,F-e) - (G/e,F/e) \Big).
    \end{equation*}
\end{itemize}  

It is important to note that in the case of \textit{topological} quantum field theories there is a direct link between graph complexes and the Feynman diagrams of their perturbative expansions. However, for ``real'' physical theories there appears to be no variant of Stokes' theorem which would allow to transfer constructions from the former to the latter case. We therefore propose here a different approach to draw a connection between the two fields.

\subsection{Philosophy}

To connect graph complexes to the study of Feynman integrals we pursue two main ideas. Our first approach continues a program initiated in \cite{BlKrCut}. It is based on the observation that the defining operations of the above mentioned complexes, collapsing or removing edges (from a spanning forest of $G$), have a natural interpretation in physics, a fact which so far has only been partially appreciated. In this regard we view edge-collapses as a means of relating the analytic structures of different -- ``neighboring" -- Feynman integrals, while removing an edge amounts to putting it \emph{on the mass-shell}, that is, to replace the corresponding propagator by its (positive energy) residue. We hence obtain applications in the study of Landau varieties of graphs and their associated monodromies -- see the next section for a list of precise results. An underlying thread is the comparison of two approaches to Feynman graphs and their analytic evaluation, the direct integration of quadrics in momentum space and the parametric approach.

Our second approach is more of an indirect nature. It is based on the observation that both graph complexes and Feynman integrals are related to various moduli spaces (of graphs). For moduli spaces of curves this is a well-known story, originating with the very work of Kontsevich that introduced graph complexes \cite{Kont1,Kont2}. For moduli spaces of graphs such complexes appear quite naturally as chain complexes associated to their cell structure \cite{rational,CoHaKaVo}. These two pictures are not unrelated, see \cite{ConantVogtmann,KarenV_MSRI} as well as \cite{ChanGalatiusPayne} which uses a moduli space of \textit{tropical curves}. 

In the world of Feynman integrals, a direct connection to moduli spaces of curves was established by the work of Francis Brown \cite{FBrown} (see also their role in the study of string scattering amplitudes \cite{francis-clement}).
Furthermore, Brown introduced canonical differential forms on moduli spaces
of metric graphs in \cite{Francis-gc}, and showed how they allow to study the cohomology of the commutative graph complex. The corresponding canonical integrals look tantalizingly similar to parametric Feynman integrals. Moreover, examples suggest that their respective periods are related by integration-by-parts methods. 
 He extended this connection to graphs with masses and kinematics recently in \cite{Francis-hgc}.

In addition, the works \cite{Marko,MaxMarko,MaxMaster} introduced moduli spaces of Feynman graphs, tailor-made to the study of Feynman amplitudes. On these spaces parametric Feynman integrals can be understood as evaluations of certain volume forms (or cochains). It is therefore natural to ask what the topology and geometry of such moduli spaces can tell us about Feynman amplitudes. 

Roughly speaking, a moduli space of graphs $\m {MG}_{n,s}$ is built as a disjoint union of cells, one for each (isomorphism class of) Feynman graph with $n$ loops and $s$ legs, glued together along face relations induced by edge collapses. This cell structure gives then rise to a graph complex via its associated chain complex $C_*(\m {MG}_{n,s})$ on which the boundary map transforms into a sum of edge-collapses (this is not quite a graph complex of Feynman diagrams, but closely related to it). Furthermore, inside this moduli space sits a homotopy equivalent subspace, called its \emph{spine}, a simplicial complex whose simplices assemble into a cube complex, parametrized by pairs $(G,F)$ where $G$ is a (Feynman) graph and $F$ a spanning forest of $G$. Its associated chain complex is the cubical chain complex described above.  
In topological terms, the former complex computes certain relative homology groups of $\m {MG}_{n,s}$ while the cubical chain complex computes its full homology. In the case of one loop graphs this relation simplifies; the spine is merely a subdivision of $\m {MG}_{1,s}$ and the two complexes are quasi-isomorphic.
\newline

The present work is to be understood as a first approximation to building a bridge between the lands of graph complexes and Feynman rules. We believe that eventually a moduli space of appropriately decorated graphs (in the sense of Culler-Vogtmann's \emph{Outer space} \cite{V}) and/or local systems on it to be the right setting to investigate the analytic structure of Feynman integrals from a geometric/topological point of view. However, already on the combinatorial level we observe how graph complexes have interesting and fruitful applications to the study of Feynman integrals. 

\subsection{Outline and results}

After setting up some notation in Sec.(\ref{Notation}) we introduce various Hopf algebras of Feynman graphs,
\begin{itemize}
 \item $H_{core}$, the Hopf algebra of core/1PI Feynman graphs, 
 \item $H_C$, a Hopf algebra of \textit{Cutkosky} graphs,
 \item $H_{GF}$, a generalization of $H_{core}$ to pairs $(G,F)$ of graphs and spanning forests.
\end{itemize}
Our first goal in Sec.(\ref{HopfAlgebras}) and (\ref{Flags}) is to define and study various maps and structures on these algebras, and to investigate how they interact with each other. To switch to the analytic side of things we recall then in Sec.(\ref{FeynmanRules}) the definition of (renormalized) Feynman rules $\Phi$ ($\Phi_R$) and in Sec.(\ref{Landau}) the notion of Landau singularities of a Feynman graph (or rather of the function defined by the integral associated to $G$ via $\Phi_R$). This sets the ground to derive the following results:

\subsubsection*{Core Graphs}
In Sec.(\ref{partialfractions}) we show that the computation $\Phi_R(G)$ of a core Feynman graph $G\in H_{core}$ can be obtained as a sum of evaluations of pairs $(G,T)$ where $T$ runs over all spanning trees of $G$ and edges not in the spanning tree are evaluated on-shell, 
\[
\Phi_R((G,T))=\sum_{\sigma\in S_{|G|}}\int_{0<s_{\sigma(|G|)}<\cdots<s_{\sigma(1)}<\infty}\left(\prod_{e\in E_T}
\frac{1}{Q_e}\right)^R_{|k(j)_0^2=s_j+m_j^2,\,j\not\in E_T}\prod_{j\not\in E_T}ds(j).
\]
See Thm.(\ref{phiGT}) for the notation. In terms of generalized Feynman rules on $H_{GF}$ this reads
\[
\Phi_R(G)=\sum_T \Phi_R((G,T)).
\]
The distinction of spanning trees upon integrating the $0$-component of loop momenta 
is also familiar in particular for one-loop graphs as a loop-tree duality, see \cite{Tomboulis} and references there. 
We use invariance properties of dimensional regularization under affine transformations of loop momenta for a systematic multi-loop approach.
We follow \cite{KreimerThesis} where a separation into parallel and orthogonal components was utilized. This separation is now systematically used by Baikov \cite{Baikov}
and leads to an interesting approach via intersection numbers \cite{Mastrolia}.
In future work we hope to connect the structure of graph complexes to 
these intersection numbers.

If one interprets Feynman amplitudes as (generalized) volumes on the moduli spaces $\m {MG}_{n,s}$ as explained above (cf.\ \cite{Marko}), then Thm.(\ref{phiGT}) shows that this point of view can also be established on the spine of $\m {MG}_{n,s}$ (recall its description as a cube complex, parametrized by pairs $(G,F)$). In other words, the moduli space is the total space of a fibration over its spine and Thm(\ref{phiGT}) is the result of integrating along its fibers (if translated into the parametric formulation). We comment on this point of view and discuss an example, leaving a detailed study to future work \cite{marko-ltd}.

\subsubsection*{Co-actions for $H_C$}
The core Hopf algebra $H_{core}$ co-acts 
\[
\bar{\Delta}_{core}: H_C\to H_{core} \otimes H_C,
\] 
on proper Cutkosky graphs $G\in H_C$ such that the computation of Feynman graphs can be reduced to a computation in $H_C^{(0)}$ and a computation in $H_{core}$. There is a direct sum decomposition 
\[
H_C=\oplus_{j=0}^\infty H_C^{(j)}
\]
where $H_C^{(j)}$ are $j$-loop graphs (and similarly for $H_{core}$), such that 
\[
\bar{\Delta}_{core}(G)=\sum_{i=0}^{j} G_{(i)}^\prime\otimes G_{(i)}^{\prime\prime}, 
\]
$G_{(i)}^\prime\in H_{core}^{(i)}$ and $G_{(i)}^{\prime\prime}\in H_C^{(j-i)}$ for $G\in H_C^{(j)}$.

From this we derive Eq.(\ref{Fubini}):
\[
\Phi_R(G)=\int\prod_{i=1}^{|{G_{(j)}^{\prime\prime}}|}d^Dk_i \left(\frac{\Phi_R(G^\prime_{(j)})}{\prod_{e\in E_{F}}Q_e}\right)_{\Big(\cap_{f\in E_{on}({G_{(j)}^{\prime\prime}})}\Big)_{(Q_f=0)}}.
\]
Here the graph $G_{(j)}^{\prime\prime}$ has edges which are off-shell ($e\in E_F$)
and their inverse product is evaluated at the loci determined by the simultaneous on-shell conditions  $Q_f=0$ for its on-shell edges. 


For the choice of a spanning tree and an ordering of edges 
$\mathfrak{o}$ we then get a sequence of such evaluations. See Sec.(\ref{usingcoaction}) for details on the co-action of the renormalization algebra.

\subsubsection*{Vanishing of the commutator $[\Delta_{GF},d_0+d_1]$}
Following the Feynman rules used in Thm.(\ref{phiGT}) only two types of edges appear:
Edges in a spanning forest $F$ remaining off-shell and edges $\not\in F$ which are evaluated on-shell.
This result implies that the co-product $\Delta_{GF}$ and pre-Lie structure  of pairs $(G,F)$ are compatible and hence commute with the boundary $d=d_0+d_1$ of the cubical chain complex, see Thm.(\ref{dCCC}).

\subsubsection*{A one-loop example}
In Sec.(\ref{doneloop}) we analyze the one-loop triangle graph and explain how it relates to a generator for the homology of the cubical chain complex furnished by the boundary $d=d_0+d_1$. Recall our interpretation, on the analytic side: $d_0$ \textit{reduces} a graph, $d_1$ puts edges on the mass-shell.

\subsubsection*{Graph homology}
In Sec.(\ref{Marko}) we consider a variant of Kontsevich's graph complex that is defined by collapsing edges in Feynman graphs. 

We show how its differential encodes which Feynman integrals share subsets of their Landau singularities. More precisely, we show that cycles represent families of graphs/integrals that ``exhaust a set of common singularities'': Each graph in the family maps under the Feynman rules to a function whose singularities are contained in a minimal common Landau variety  (cf.\ Thm.(\ref{prop:cycles}) for a precise definition of this property). 

For a theory with cubic interaction this gives a direct connection between the top dimensional graph homology group and the analytic structure of Feynman amplitudes. In the one loop case the elements of the homology classes induce a nice partition of the set of graphs contributing to the full amplitude. Each subset of this partition satisfies the above mentioned property of sharing singularities while also obeying certain symmetry relations. We prove this and comment on extensions  in Sec.(\ref{ss:partitiongreensfctn}).

\subsection*{Acknowledgments} We thank Karen Vogtmann for many valuable comments on a first draft of this paper. DK thanks Spencer Bloch 
for an uncountable number of insightful conversations 
on the mathematical structure of Cutkosky rules. He also thanks Karen Yeats for a longstanding collaboration on combinatorial aspects of Feynman amplitudes and Michael Borinsky for discussions on Feynman amplitudes.
Finally, DK wants to thank Adrian Roosch
and Patricia Schr\"oder for exercising miracles through physiotherapy. MB thanks Paul Balduf and Erik Panzer for helpful comments. In addition, he thanks Max M\"uhlbauer for numerous valuable discussions while sending very different types of problems.

\section{Graphs, spanning trees, refinements}\label{Notation}
Note that our definition of graphs closely follows the set-up of \cite{Kr-Y}.
We first settle the notion of a partition.
\begin{defn}
Given a set $S$ a \emph{partition} (or \emph{set partition}) $\mathcal{P}$ of $S$ is a decomposition of $S$ into disjoint nonempty subsets whose union is $S$.  The subsets forming this decomposition are the \emph{parts} of $\mathcal{P}$.  The parts of a partition are unordered, but it is often convenient to write a partition with $k$ parts as $\dot{\cup}_{i=1}^k S_i = S$ with the understanding that permuting the $S_i$ still gives the same partition.  A partition $\mathcal{P}$ with $k$ parts is called a $k$-partition and we write $k=|\mathcal{P}|$.
\end{defn}
Now we can define a Feynman graph.
\begin{defn}
A \emph{Feynman graph} $G$ is a tuple $G=(H_G, \mathcal{V}_G, \mathcal{E}_G)$ consisting of 
\begin{itemize}
\item $H_G$, the set of half-edges of $G$,
\item $\mathcal{V}_G$, a partition of $H_G$ with parts of cardinality at least 3 giving the vertices of $G$,
\item $\mathcal{E}_G$, a partition of $H_G$ with parts of cardinality at most 2 giving the edges of $G$.
\end{itemize}
\end{defn}
From now on when we say graph we mean a Feynman graph.

We do not require all parts of $\mathcal{E}_G$ to be of  cardinality 2.
We identify the parts of cardinality 2 with the set of edges $E_G$ of the graph and set $e_G:=|E_G|$.  We identify the sets of cardinality 1 with the set of external edges $L_G$ of the graph and set $l_G:=|L_G|$. Also we set $v_G:=|\mathcal{V}_G|$.

We say that a graph $G$ is connected if there is no partition of the parts of $\mathcal{V}_G$ into two sets $H_G(1),H_G(2)$ such that the parts of cardinality two of $\mathcal{E}_G$
are either in $H_G(1)$ or $H_G(2)$. If it is not connected it has $|H^0(G)|>1$ components.

The partition $\mathcal{V}_G$ collects half-edges of $G$ into vertices.  This formulation of graphs does not distinguish between a vertex and the corolla of half-edges giving that vertex.  However, it is sometime useful to have notation to distinguish when one should think of vertices as vertices and when one should think of them as corollas.  Consequently let $V_G$, the set of vertices of $G$, be a set in bijection with the parts of $\mathcal{V}_G$, $|V_G|=v_G=|\mathcal{V}_G|$.  This bijection can be extended to a map $\nu_G:H_G\rightarrow V_G$ by taking each half edge to the vertex corresponding to the part of $\mathcal{V}_G$ containing that vertex.
For $v\in V_G$ define 
\[
C_v:=\nu_G^{-1}(v)\subset H_G,
\]
to be the corolla at $v$, that is the part of $\mathcal{V}_G$ corresponding to $v$. 

A graph $G$ as above can be regarded as a set of corollas determined by $\mathcal{V}_G$ glued together according to $\mathcal{E}_G$.

If $|\nu_G(e)|=1$, we say $e$ is a self-loop at $v$, with $\nu_G(e)=\{v\}$.

We frequently have cause to make an arbitrary choice of an orientation on the edges. 
If $|\nu_G(e)|=2$, with $e=\{l,m\}$ and $\nu(l)=v,\nu(m)=w$ say, $e$ is an edge 
$e_{vw}$ from $v$ to $w$ or $e_{wv}$ vice versa for the opposite orientation. This choice of an edge orientation corresponds to a choice of an order of $e$ as a set of half-edges.

If we orient an edge $e$, we also write $v_+(e_{vw})=w$ and $v_-(e_{vw})=v$ for the source and target vertices.

We emphasize that we allow multiple edges between vertices and allow self-loops as well. 

We write $h_1(G)\equiv |G|:=|H^1(G)|=e_G-v_G+|H^0(G)|$ for the number of independent loops (cycles), or the dimension of the cycle space of the graph $G$.
Note that for disjoint unions of graphs $h_1,h_2$, we have $|h_1\cup h_2|=|h_1|+|h_2|$. We write $h_0(G):=|H^0(G)|$.

A graph is bridgeless if $(G-e)$ has the same number of connected components as $G$ for any $e\in E_G$.  A graph is 1PI or 2-edge-connected if it is both bridgeless and connected, equivalently if $(G-e)$ is connected for any $e\in E_G$. 
Here,
for $G=(H_G, \mathcal{V}_G, \mathcal{E}_G)$, we define
\[
(G-e):= (H_G, \mathcal{V}_G, \mathcal{E}'_G)
\]
where $\mathcal{E}'_G$ is the partition which is the same as $\mathcal{E}_G$ except that the part corresponding to $e$ is split into two parts of size $1$.

The removal $G-X$ of edges forming a subgraph $X\subset G$ is defined similarly 
by splitting the parts of $\mathcal{E}_G$ corresponding to edges of $X$.
$G-X$ can contain isolated corollas.

Note that this definition is different from graph theoretic edge deletion as all the half-edges of the graph remain and the corollas are unchanged.  
We neither lose vertices nor half-edges when removing an internal edge. We just unglue
the two corollas connected by that edge.

The graph resulting from the contraction of edge $e$, denoted $G/e$ for $e\in E_G$, is defined to be
\be\label{edgecont}
G/e = (H_G-e, \mathcal{V}'_G, \mathcal{E}_G-e)
\ee
where $\mathcal{V}'_G$ is the partition which is the same as $\mathcal{V}_G$ except that in place of the parts $C_v$ and $C_w$ for $e=\{\nu^{-1}(v),\nu^{-1}(w)\}$, $\mathcal{V}'$ has a single part $(C_v\cup C_w) - e$.\footnote{We often use $-$ for the set difference, e.g.\ $H_G-e=H_G\setminus e$.}

Likewise we define $G/X$, for $X\subseteq G$ a (not necessarily connected) graph, to be the graph obtained from $G$ by contracting all internal edges of $X\subseteq G$.

Intuitively we can think of $G/X$ as the graph resulting by shrinking all internal edges of $X$ to zero length:
\be\label{lengthcont}
G/X=G|_{\mathrm{length}(e)=0,e\in E_X}.
\ee
This intuitive definition can be made into a precise definition if we add the notion of edge lengths to our graphs, but doing so is not to the point at present.

Note that restricting $\mathcal{V}_G$ to $L_G$ we also obtain a partition of $L_G$ into the sets $L_G\cap \nu_G^{-1}(v)$:\footnote{Technically we must discard any subsets which are now empty in order to obtain a partition.}
\[
L_G=\dot{\cup}_{v\in V_G} \underbrace{\left(L_G\cap \nu_G^{-1}(v)\right)}_{=:L_v}.
\]
We let $\mathbf{val}(v):=|C_v|$ the degree or valence of $v$ and $\mathbf{eval}(v):=|L_v|$
the number of external edges at $v$, and $\mathbf{ival}(v):=\mathbf{val}(v)-\mathbf{eval}(v)$ the number of internal edges at $v$.

Summarizing, for a graph $G$ we have an  internal edge set $E_G$, vertex set $V_G$ and  set of external edges $L_G$.

A simply connected subset of edges $T$ which contains $V_G$
we call a spanning tree of $G$. For any proper subset $f$ of
edges of  $T$ we call $F=T-f$ a spanning forest of $G$. Note that a spanning forest of $G$ contains all vertices of $G$.

It induces a graph $(H_G,\mathcal{V},\mathcal{F})$ on the same set of half-edges and vertices as $G$, and with a refined edge partition $\mathcal{F}$ defined by retaining as parts of cardinality two only the edges of $F$.

We often notate this as a pair $(G,F)$. We also write $G_F$
for such a pair. By a Cutkosky graph we mean such a pair \cite{Kr-Y}. 

The set of edges $e\in E_G$ such that $e\not\in E_F$ forms the set $E_{on}$ of $G$, the set of edges $e\in E_F$ the set
$E_{off}$. Note that $G_T$ has a non-empty set $E_{on}$,
$|E_{on}|=|G|$.

Any spanning tree $T$ is also a forest with $F=\emptyset$ such that $|E_{on}|=|G|$ as $E_{on}$ provides a basis for the loops $l_e\in\mathcal{L}$ of $G$: for any $e\in E_{on}$, there is a path $p_e\subset T$ such that $l_e=e\dot{\cup} p_e$ is a loop.

\begin{defn}
Given two partitions $\mathcal{P}$ and $\mathcal{P}'$ of a set $S$, we say $\mathcal{P}'$ is a \emph{refinement} of $\mathcal{P}$ if every part of $\mathcal{P}'$ is a subset of a part of $\mathcal{P}$.  Intuitively $\mathcal{P}'$ can be made from $\mathcal{P}$ by splitting some parts.  The set of all partitions of $S$ with the refinement relation gives a lattice called the \emph{partition lattice}. The covering relation in this lattice is the special case of refinement where exactly one part of $\mathcal{P}$ is split into two parts to give $\mathcal{P}'$.

We will need more than just the refinements of partitions as defined above.  Given a refinement $\mathcal{P}'$ of $\mathcal{P}$ it will often be useful that we additionally pick a maximal chain from $\mathcal{P}$ to $\mathcal{P}'$ in the partition lattice.  Concretely this means we keep track of a way to build $\mathcal{P}'$ from $\mathcal{P}$ by a linear sequence of steps, each of which splits exactly one part into two.  Unless otherwise specified our refinements always come with this sequence building them, and we will let a \emph{$j$-refinement} be such a refinement where the sequence $\mathcal{P}(i),0\leq i\leq j$ of partitions has length $j$ (including both ends). $\mathcal{P}(0)=S$ is the trivial partition.
\end{defn}

An ordering $\mathfrak{o}$ of the edges in a spanning tree defines a $v_G$-refinement of $G$ with corresponding refinement of $L_G$.\footnote{A removal of edges from the spanning tree in any order induces a removal of edges from the graph which connect different
components of the resulting  spanning forest giving a corresponding refinement of the graph.} 

We define the vectorspace $H_{core}$ as the $\mathbb{Q}$ vectorspace generated by (disjoint unions of) bridgeless connected (core)  graphs $G$.

Similarly we define the $\mathbb{Q}$ vectorspace $H_{GF}$ generated by (disjoint unions of) pairs of a core graph $G$ and spanning tree $F$ of $G$.

Finally we define the $\mathbb{Q}$ vectorspace $H_{C}$ generated by 
(disjoint unions of) Cutkosky graphs.

\section{Hopf algebras}\label{HopfAlgebras}
Again, our set-up is closely related to \cite{Kr-Y}. We first define the Hopf algebras $H_{core}$. It will co-act on $H_C$ defined above. $H_{core}$ is central in studying the relation between quantum fields and the structure of Outer Space, see 
\cite{Kr-Y} and also \cite{BorVogtmann}.
\subsection{The core Hopf algebra $H_{core}$}
The core Hopf algebra $H_{core}$ \cite{core,BV} is based on the $\mathbb{Q}$-vectorspace generated by connected bridgeless Feynman graphs and their disjoint unions.

We define  a commutative product
\[
m: H_{core}\otimes H_{core}\to H_{core},\, m(G_1,G_2)=G_1\dot{\cup} G_2,
\]
by disjoint union. The unit $\One$ is provided by the empty set so that we get a free commutative $\mathbb{Q}$-algebra with bridgeless connected graphs as generators.

We define a co-product by
\[
\Delta_{core}(G)=G\otimes\One+\One\otimes G+\sum_{g\subsetneq G}g\otimes G/g,
\]
where the sum is over all $g\in H_{core}$ such that $g\subsetneq G$.
Hence there are bridgeless graphs $g_i$ such that $g=\dot{\cup}_i g_i$, and $G/g$ denotes the co-graph in which all internal edges of all $g_i$ shrink to zero length in $G$. We define the reduced co-product to be
\[
\tilde{\Delta}_{core}(G)=\sum_{g\subsetneq G}g\otimes G/g,
\]

We have a co-unit $\hat{\One}:H_{core}\to\mathbb{Q}$ which annihilates any non-empty graph and $\hat{\One}(\One)=1$ and we have the antipode
$S:H_{core}\to H_{core}$, $S(\One)=\One$
\[
S(G)=-G-\sum_{g\subsetneq G}S(g) G/g.
\]
Furthermore our Hopf algebras are graded,
\[
H_{core}=\oplus_{j=0}^\infty H_{core}^{(j)},\,H_{core}^{(0)}\cong \mathbb{Q}\One,\, \mathrm{Aug}_{core}=\oplus_{j=1}^\infty H_{core}^{(j)},
\]
and $h\in H_{core}^{(j)}\Leftrightarrow |h|=j$.
The core Hopf algebra has various quotient Hopf algebras amongst them the Hopf algebra for renormalization $H_{ren}$, see \cite{BV}. 
\subsection{The Hopf algebra $H_{GF}$}
The Hopf algebra $H_{core}$ has a generalization $H_{GF}$ operating on pairs $(G,F)$ of a graph $G$ and a spanning forest $F$ \cite{Kr-Y}.

 Let $\mathcal{F}_G$ be the set of all
spanning forests of $G$. It includes the set $\mathcal{T}_G$ of all spanning trees of $G$.  The empty graph $\One$ has an empty spanning forest also denoted by $\One$.

Each spanning tree $T$ of $G$ gives rise to a set of cycles $\mathcal{L}=\mathcal{L}(T)$.

The powerset $\dot{\cup}_T \mathcal{P}_T$ of these cycles can be identified with the set of
all subgraphs of $(G,T)$. 
 
Each forest $F$ defines a partition $L_G(F)$ of the set of external edges of $G$. In fact for two pairs $(G;F),(G^\prime,F^\prime)$ with the same set of external edges $L_G=L_{G^\prime}$
we say $F\sim F^\prime$ if they define the same partition: 
\[
L_G(F)=L_{G^\prime}(F^\prime).
\]  

 We define a $\mathbb{Q}$-Hopf algebra $H_{GF}$ for such pairs $(G,F)$, $F\in \mathcal{F}_G$ by setting 
\bea\label{HopfPairs}
\Delta_{GF}(G,F) & = & (G,F)\otimes (\One,\One)+(\One,\One)\otimes (G,F)+\nonumber\\
 &  & +\sum_{{g\subsetneq G \atop F-(F\cap g)\in \mathcal{F}_{G/g}} \atop
 F\sim F-(F\cap g)} (g,g\cap F)\otimes (G/g, F-(F\cap g)),
\eea
 where $\mathcal{F_G}$ is the set of all forests of $G$.  Additionally, by $F-(F\cap g)\in \mathcal{F}_{G/g}$ we mean to interpret the edges of $F-(F\cap g)$ as a subgraph of $G/g$ and then check if that subgraph is an element of $\mathcal{F}_{G/g}$.  This ensures that 
only terms contribute such that $G/g$ has a valid spanning forest.  Finally, 
by $F\sim F-(F\cap g)$ we mean that the partition of external legs 
of $(G,F)$ and $(G/g,F-(F\cap g))$ are identical.

We define the commutative product to be
\[
m_{GF}((G_1,F_1),(G_2,F_2))=(G_1\dot{\cup} G_2,F_1\dot{\cup} F_2),
\]
whilst $\One_{GF}=(\One,\One)$ serves as the obvious unit which induces a co-unit 
through $\hat{\One}_{GF}(\One_{GF})=1$ and $\hat{\One}_{GF}((G,F))=0$. 

\begin{thm}
This is a graded commutative bi-algebra graded by $|G|$
and therefore a Hopf algebra $H_{GF}(\One_{GF},\hat{\One}_{GF},m_{GF},\Delta_{GF},S_{GF})$.
\end{thm}
\begin{proof}
We rely on the co-associativity of $H_{core}$ which holds for graphs with labeled edges.
Using Sweedler's notation this amounts to
\bea
& & \sum_{i,j}(G^\prime_{(i)})^\prime_{(j)}\otimes
(G^\prime_{(i)})^{\prime\prime}_{(j)}\otimes
G^{\prime\prime}_{(i)}\nonumber\\
& = &  \sum_{i,j} G^{\prime}_{(i)} \otimes
(G^{\prime\prime}_{(i)})^{\prime}_{(j)}\otimes
(G^{\prime\prime}_{(i)})^{\prime\prime}_{(j)}
\eea
for any graph $G$.
Consider all edges $e\in E_F$ as labeled. The core co-product 
generates loops in these labeled edges in its first application only in the right slot, and when applying it again at most in the two slots to the right. 
We have to show that the same terms are eliminated when we abandon terms with loops from eges in $E_F$ respecting co-associativity.

The assertion follows:\\
iff $G^\prime_{(i)}/(G^\prime_{(i)})^\prime_{(j)}$ contains a loop
then  $G^{\prime\prime}_{(i)}/(G^\prime_{(i)})^\prime_{(j)}$
contains that loop and\\
iff $G^{\prime\prime}_{(i)}$ contains a loop then
either $(G^{\prime\prime}_{(i)})^\prime_{(j)}$ or
$(G^{\prime\prime}_{(i)})^{\prime\prime}_{(j)}$ will.
\end{proof}
We have $H_{GF}=\oplus_{j=0}^\infty H_{GF}^{(j)}$ with $H_{GF}^{(0)}\sim\mathbb{Q}\One_{GF}$ and $\mathrm{Aug}_{GF}=\oplus_{j=1}^\infty H_{GF}^{(j)}$.
$(G,F)\in H_{GF}^{(j)}\Leftrightarrow |G|=j,\, F\in\mathcal{F}_G$.
\subsection{The vectorspace $H_C$}
Consider a Cutkosky graph $G$ with a corresponding $v_G$-refinement $P$ of its set  of external edges $L_G$. It is a maximal refinement of $V_G$ corresponding to the choice of an ordered spanning tree.

The core Hopf algebra co-acts on the vector-space of Cutkosky graphs $H_{C}$.
\be
\bar{\Delta}_{core}: H_C\to H_{core}\otimes H_C,\,\bar{\Delta}_{core}(G)=\One\otimes G+\sum_{g\subsetneq G,\,g\in H_{core}}g\otimes G/g. 
\ee
We set $G\in H_C^{(n)}\Leftrightarrow |G|=n$ and decompose $H_C=\oplus_{i=0}^\infty H_C^{(i)}$. 

Note that the sub-vectorspace $H_C^{(0)}$ is rather large: it contains all Cutkosky graphs $G=((H_G,\mathcal{V}_G,\mathcal{F}_G)$ such that $|G|=0$. These are the graphs where the cuts leave no loop intact.

For any $G\in H_C$ there exists a largest  integer $\mathrm{cor}_C(G)\geq 0$ such that 
\[
\tilde{\bar{\Delta}}_{core}^{\mathrm{cor}_C(G)}(G)\not= 0,\, \tilde{\bar{\Delta}}_{core}^{\mathrm{cor}_C(G)}(G):H_C\to H_{core}^{\otimes cor_C(G)}\otimes H_C^{(0)}, 
\]
whilst $\tilde{\Delta}_{core}^{\mathrm{cor}_C(G)+1}(G)= 0$.
\begin{prop}
\[
\mathrm{cor}_C(G)=|G|.
\]
\end{prop}
\begin{proof}
The primitives of $H_{core}$ are one-loop graphs.
\end{proof}

As $\bar{\Delta}_{core}:\,H_C\to H_{core}\otimes H_C$ there is for any $G\in H_C$ a unique $g\in H_{core}$ such that $G/g\in H_C^{(0)}$ has no loops.
\begin{cor}\label{gG0}
There is a unique element $g\otimes G/g\in H_{core}\otimes H_C^{(0)}$:
\[
\bar{\Delta}_{core}(G)\cap \left(H_{core}\otimes H_C^{(0)}\right)=g\otimes G/g, 
\]
with $|g|=|G|$.
\end{cor}

\section{Flags} \label{Flags}
The notion of flags of Feynman graphs was for example already used in \cite{BlKrLMHS,PanzerHepp}.
Here we use it based on the core Hopf algebra introduced above.

\subsection{Expanded flags}
Consider a graph $G$. We define as an expanded  flag associated to $G$ a sequence of graphs 
\[
\tilde{f}:=G_1\subsetneq G_2\subsetneq \cdots \subsetneq G_{|G|}=G,
\]
where $|G_1|=1$ and $|G_i/G_{i-1}|=1$ for all $i\geq 2$. We set $\gamma_i:=G_i/G_{i-1}$ and $\gamma_1:=G_1$.

 Write
$\mathcal{F}l(G)$ for the collection of all expanded flags $\tilde{f}\in \mathcal{F}l(G)$  of $G$. 

\subsection{Flags}
The flag $f\in \mathrm{Aug}_{core}^{\otimes k}$ of length $|G|$ associated to $\tilde{f}$ is 
\[
f:=\gamma_1\otimes \cdots\otimes \gamma_{|G|}.
\]

Define the flag associated to a graph $G\in \langle H_{core}\rangle$ to be a sum of flags of length $|G|$ arising from all expanded flags:
\[
Fl_G:=\sum_{ \tilde{f} \in \m{F}l(G) } f=\tilde{\Delta}_{core}^{|G|-1}(G).
\]

With $\xi_G=|F(G)|$ the number of expanded flags a graph $G$ has we can hence write 
\[
Fl_G=\sum_{i=1}^{\xi_G} \gamma_1^{(i)}\otimes\cdots\otimes \gamma_{|G|}^{(i)},
\]
where for any of the  orderings of the cycles $l_j$ of $G$ we have 
\be\label{flagg} 
\gamma_1=l_1, \gamma_2=l_2/E_{l_1\cap l_2}, \ldots, \gamma_{|G|}=l_G/E_{l_1\cap\cdots\cap l_{|G|-1}}.
\ee

Similarly, for a pair $(G,F)$ we can define
\[
Fl_{G,F}:=\tilde{\Delta}_{GF}^{|G|-1}((G,F))\in \mathrm{Aug}_{GF}^{\otimes |G|},
\]
which as a sum of flags is 
\[
Fl_{G,F}=\sum_i(\gamma_1,f_1)^{i}\otimes\cdots \otimes (\gamma_{|G|},f_{|G|})^{i},
\]
in an obvious manner. Here
$\tilde{\Delta}_{GF}((\gamma_l,f_l)^i)=0$, $\forall i,l$, $1\leq l\leq |G|$.

See Fig.(\ref{flagtheta}) for an example.
\begin{figure}[H]
\includegraphics[width=12cm]{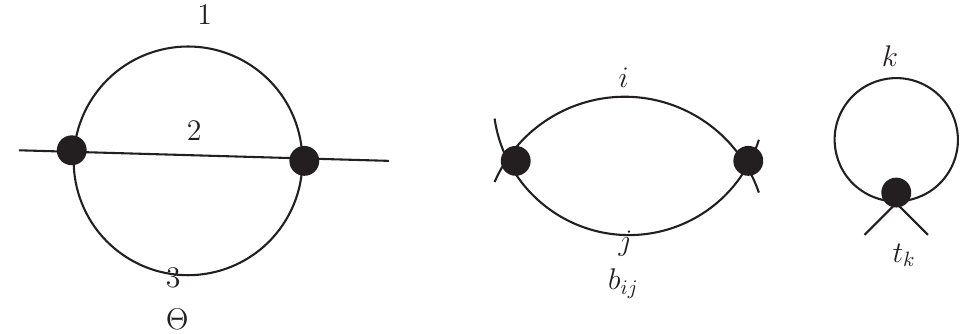}.
\caption{The 3-edge banana graph $\Theta$ on edges $e_1,e_2,e_3$. It has three 2-edge subgraphs $b_{ij}$ on edges $e_i,e_j$, with a cograph $t_k$ on edge $e_k$.
$Fl_\Theta=\gamma_{12}\otimes t_3+\gamma_{23}\otimes t_1+\gamma_{31}\otimes t_2$. The three cycles in $\theta$ are $l_1=e_1,e_2$, $l_2=e_2,e_3$ and $l_3=e_3,e_1$. If $l_1,l_2$ are chosen as a basis (so $e_2$ is the spanning tree) in the order $l_1<l_2$ we have
$t_3=l_2/E_{l_1\cap l_2}$. $Fl_{\Theta,e_2}=\gamma_{12}\otimes t_3+\gamma_{23}\otimes t_1$. With three spanning trees and two orders we thus get six terms.
In fact the $\gamma_{ij}$ subgraphs have two residues, the tadpoles have one, so that we have a decomposition $6=2\times 1+2\times 1+2\times 1$ into $3=\xi_\Theta$ parts.}
\label{flagtheta}
\end{figure}
\subsection{Flags for ordered spanning trees}
We consider pairs of a graph $G$ and a spanning tree $T$ but this time we assume that there is an order $\mathfrak{o}=\mathfrak{o}(T)$ on the edges of the spanning tree $T$.

For any decomposition $T=T^\prime \cup T^{\prime\prime}$ into two disjoint subtrees,
we say that the pair $(T^\prime,T^{\prime\prime})$ is $\mathfrak{o}$-compatible,
\[
(T^\prime,T^{\prime\prime})\sim \mathfrak{o},
\]
if any edge $e\in T^\prime$ is ordered before any edge  
$f\in T^{\prime\prime}$, so that  $\mathfrak{o}$ is a concatenation
\[
\mathfrak{o}(T^\prime)\mathfrak{o}(T^{\prime\prime}).
\]

We now define a map by restricting $\Delta_{GF}$ for any order $\mathfrak{o}$ to
$\mathfrak{o}$-compatible terms,
\bea\label{HopfPairso}
\Delta_{GF}^{\mathfrak{o}}(G,T) & := & (G,T)\otimes (\One,\One)+(\One,\One)\otimes (G,T)+\nonumber\\
 &  & +\sum_{
 {g\subsetneq G \atop T-(T\cap g)\in \mathcal{F}_{G/g},\, 
 {T\sim T-(T\cap g)}} 
 \atop 
 {(g\cap T,T-(T\cap g))\sim\mathfrak{o}}
 } (g,g\cap T)\otimes (G/g, T-(T\cap g)).
\eea
The generalization of this map to ordered forest instead of ordered spanning trees is straightforward  as the definition in Eq.(\ref{HopfPairs}) ensures compatibility of cuts on graphs and co-graphs \cite{Kr-Y}. We use $\Delta_{GF}^{\mathfrak{o}}(G,T)$  later to investigate a Leibniz rule apparent in the cubical chain complex in Sec.(\ref{prelieandccc}).
\begin{rem}
Maps $\Delta_{GF}^{\mathfrak{o}}(G,T)$ can be made co-associative in any sector as described by $\mathfrak{o}$. In each sector they give rise to a co-product on decorated rooted trees without side-branchings \cite{BlKrLMHS}. Their connection to co-interacting bi-algebra structures and sector decompositions (see \cite{Kr-Y}) is left for future work.
\end{rem}
 
We now give an example of such a map. We choose a pair $(G,T)$ with $T$ with a spanning tree of length three. We label its edges by $a,b,c$ and choose the order $acb$. $|G|=3$ and the edges not in the spanning tree -labeled $1,2,3$- define
a base for the fundamental cycles of the graph.
\begin{figure}[H]
\includegraphics[width=12cm]{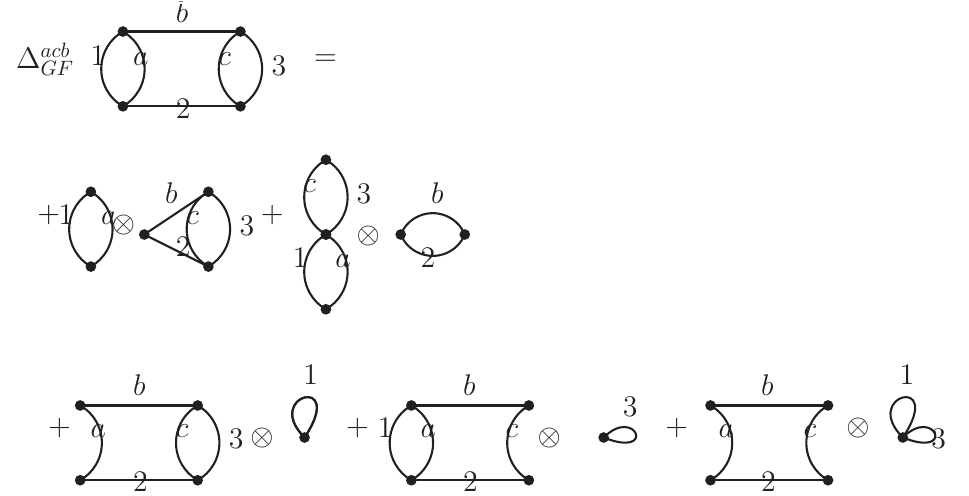}.
\caption{The map $\Delta_{GF}^{\mathfrak{o}=acb}$ in an example. Note that the three terms in the lowest row all have tadpoles on the rhs. They will always be present for any order order we choose. In particular the maps $\Delta_{GF}^{\mathfrak{o}=bac}$ and
$\Delta_{GF}^{\mathfrak{o}=bca}$ produce just those three terms in the lowest row.}
\label{deltagfo}
\end{figure}

\section{Feynman rules} \label{FeynmanRules}
\subsection{Momentum space renormalized Feynman rules}
Consider a graph $G$ with set of external half-edges $L_G$. All external half-edges are oriented incoming.

To each $f\in L_G$ assign an external momentum $q(f)\in \mathbb{M}^D$.

Next, choose an orientation for each edge $e\in E_G$ and assign an internal momentum 
$k(e)\in \mathbb{M}^D$ to each edge. With these orientations the half-edges 
$h\in C_v$ at a vertex $v$ are oriented. We say that $k(e)$ is incoming at $v$ if $h\in e$ is oriented towards $v$ ($v$ is the target of $e$). 
We set $k(h)=k(e)$. Else, if $v$ is the source of $e$, $-k(e)$ is incoming at $v$. We set $k(h)=-k(e)$.

Define the integral
\begin{equation}\label{eq:momfi}
 I_G(\{q(f)\},\{m_e\}):=\int_{\mathbb{M}^{De_G}} d^{De_G}k\prod_{e\in E_G}\frac{1}{k_e^2-m_e^2+i\eta}\prod_{v\in V_G}\delta^{(D)}\sum_{h\in C_v}k(h).
 \end{equation}
By momentum conservation at each vertex this is a $D\times |G|$-dimensional integral.

Imposing kinematic renormalization conditions the renormalized integral is given as
\[
I_G^R(\{q(f)\},\{\mu(f)\},\{m_e\})=\lim_{D \to 4}\left(\sum I_{S(G^\prime)}\left(\{\mu(h)\},\{m_e\}\right)\times I_{G^{\prime\prime}}\left(\{q(f)\},\{m_e\}\right)\right)
\]
using Sweedler's notation for the coproduct $ \Delta_{ren}(G)=\sum G^\prime\otimes G^{\prime\prime}$ of the renormalization Hopf algebra $H_{ren}$ 
and a kinematic renormalization scheme which subtracts on the level of the integrand. $S$ is the antipode of $H_{ren}$. $H_{ren}$ is the usual quotient of $H_{core}$ obtained by discarding superficially convergent diagrams \cite{Kr-Y}. 
\subsection{Renormalized quadrics}
The integrands above are products of quadrics (taking momentum conservation at each vertex into account)
\be\label{integrandq}
I_G^\Pi (\{q(f),\{m_e\}\}):=\prod_{e\in E_G} \frac{1}{Q_e}.
\ee
The renormalized integrand is then
\be\label{integrandqren}
I_G^R(\{q(f)\},\{\mu(f),\{m_e\}\})=I_{S(G^\prime)}^\Pi (\{\mu(h)\},\{m_e\})I_{G^{\prime\prime}}^\Pi (\{q(f)\},\{m_e\}).
\ee
\subsection{Symanzik polynomials}
Let $\psi(G),\phi(G)$ be the two usual graph polynomials, and 
\be \Xi(G)=\phi(G)-M(G)\psi(G),\ee
the full second graph polynomial with masses.
Here,
\be 
M(G):=\sum_{e\in E_G}m_e^2 a_e,
\ee
and it is understood that all $i\eta$'s are absorbed in the masses $m_e$.
We have 
\be
\psi(G)=\psi(G/\gamma)\psi(\gamma)+R_\gamma^G,
\ee
\be
\phi(G)=\phi(G/\gamma)\psi(\gamma)+\tilde{R}_\gamma^G.
\ee
\be
\Xi(G)=\Xi(G/\gamma)\psi(\gamma)+\bar{R}_\gamma^G.
\ee
\be
\psi(G_1G_2)=\psi(G_1)\psi(G_2), 
\ee
\be
\phi(G_1G_2)=\phi(G_1)\psi(G_2)+\phi(G_2)\psi(G_1), 
\ee
\be
\Xi(G_1G_2)=\Xi(G_1)\psi(G_2)+\Xi(G_2)\psi(G_1).
\ee
Here, the remainders $R_\gamma^G$, $\tilde{R}_\gamma^G$, $\bar{R}_\gamma^G$
are all of higher degrees in the subgraph variables than $\psi(\gamma)$. This is crucial to achieve renormalizability \cite{BrownKreimer}.
 
\subsection{Parametric renormalized Feynman rules}\label{parFR}

Omitting constant prefactors and absorbing the $i\eta$'s into the masses $m_e$, the parametric version of the Feynman integral \eqref{eq:momfi} reads
\begin{equation*}
 \int_{\mb P_G} \psi_G^{-\frac{D}{2}} \left( \frac{\psi_G}{\Xi_G} \right)^{w(G)} \Omega_G.
\end{equation*}
Here $w$ denotes the superficial degree of divergence, $\mb P_G$ is the standard projective simplex,
\begin{equation*}
\mb P_G := \mb P\big(\mb R_{\geq 0}^{e_G}\big)= \{ [a_1:\cdots:a_{e_G}] \mid a_i \geq 0  \} \subset \mb {CP}^{e_G-1},
\end{equation*}
and 
\begin{equation*}
 \Omega_G= \sum_{i=1}^{e_G} (-1)^i a_i da_1 \wedge \ldots \wedge \widehat{da_i} \wedge \ldots \wedge da_{e_G}.
\end{equation*}

In a renormalizable field theory, we then get renormalized Feynman rules for an overall logarithmically divergent graph $G$ ($w(G)=0$) 
with logarithmically divergent subgraphs as
\be 
\Phi_R(G)=\int_{\mathbb{P}_G}\sum_{F\in\mathcal{F}_G}(-1)^{|F|}\frac{\ln\frac{\Xi_{G/F}\psi_F+\Xi^0_F\psi_{G/F}}{\Xi^0_{G/F}\psi_F+\Xi^0_F\psi_{G/F}}}{\psi^2_{G/F}\psi^2_F} \Omega_G.
\ee

Formulae for other degrees of divergence for sub- and cographs and further details can be found in \cite{BrownKreimer}. In particular, also overall convergent graphs are covered. 
\subsection{Cut graphs}\label{cutgaphs}
We now give the Feynman rules for graphs $\in H_C$. This can be regarded as giving Feynman rules for a pair $(G,F)$.
\be
\Upsilon_G^F:=\int\left(\prod_{e\in E_F}\frac{1}{P(e)}\prod_{e\not\in E_F}\delta^+(P(e))\right)d^{4e_G}k.\label{Upsilon}
\ee
This holds when the pair $(G,F)$ does not require renormalization. Else we proceed using the co-action of $H_{ren}$ induced by $H_{core}$ on $H_C$, see Sec.(\ref{usingcoaction}) in accordance with the above.

\section{Landau singularities}\label{Landau}
We give a quick recap of Landau singularities of Feynman integrals. For a detailed treatment we refer to the standard textbook \cite{ELOP}, for a short account to the classic paper \cite{colenort}. A mathematical rigorous discussion can be found in \cite{Pham}.

A Feynman graph $G$ represents via the above introduced (renormalized) Feynman rules $\Phi_R$ a function $\Phi_R(G)$ of its kinematics, that is, the external momenta and internal masses. 
If we restrict the allowed masses to a finite set, then each bare graph represents a finite family of such functions, parametrized by the distribution of masses on its internal edges. We model this family by edge-colorings of $G$ where the set of colors $C\subset \mb N$ represents the mass spectrum. In the following we let $G=(G,c)$, with $c:E_G\to C$ the coloring map, always denote a colored graph.

A classic result\footnote{Strictly speaking, the statement is classic, but not the result. Astonishingly, there does not exist a rigorous derivation in the published literature. See the recent works \cite{collinsletter} and \cite{maxnewpaper} for a discussion.} establishes the analyticity of $\Phi_R(G)$ outside an analytic set in the space of kinematic invariants, the \textit{Landau variety} $\mathbb{L}_G$ of $G$. More precisely, the analytic set of singularities of $\Phi_R(G)$ is a subset of $\mathbb{L}_G$ since its equations give only necessary, but not sufficient conditions for $\Phi_R(G)$ to exhibit a singularity.

Using the ``Feynman trick'' 
\begin{equation*}
 \prod_{i=1}^n \frac{1}{X_i}=\int_{[0,1]^n} da_1 \cdots da_n \frac{\delta_0(1-\sum_{i=1}^n a_i)}{(a_1X_1 + \ldots + a_nX_n)^n} =\int_{\Delta^{n-1}}  \frac{da_1 \cdots da_n}{(a_1X_1 + \ldots + a_nX_n)^n}
\end{equation*}
 we may rewrite a momentum space Feynman integral (omitting the $\delta$ factors) as
\begin{equation*}\label{eq:feynman}
 \Phi(G)=\int_{\mb M^{De_G}} dk \prod_{e \in E_G} \frac{1}{Q_e}= \int_{\mb M^{De_G}} dk \int_{\mb P_G} da  \Big( \sum_{e\in E_G} Q_e a_e \Big)^{-e_G}. 
\end{equation*}

From this we find the poles of the integrand characterized by 
\begin{equation}
\tag{L1}
  \forall e\in E_G: \text{ either } k_e^2=m_e^2 \text{ or } a_e=0 .
  \label{eq:L1}
\end{equation}

Some of these poles might still be integrable by a suitable deformation of the integration contour. Such a deformation is impossible if either the contour of integration gets ``pinched'' by the singular hypersurface specified by the equation above or if it occurs at a boundary point of $\mb P_G$, that is, for some $a_e=0$. The pinching condition translates to 
\begin{equation}
\tag{L2}
 \sum_{e \in E_l} a_e k_e =0\text{ for each simple loop $l$ in $G$. }
\label{eq:L2}
 \end{equation}

These two conditions constitute the \textit{Landau equations}, their solution set defines the Landau variety $\mathbb{L}_G$. Some authors include the side constraint $a_e \geq 0$, in other conventions it is used to distinguish \textit{physical} from \textit{non-physical} singularities.

Note that the first set of Landau equations \eqref{eq:L1} induces a natural stratification of $\mb L_G$. Moreover, the strata inherit a partial order from the boundary structure of $\mb P_G$.

\begin{defn}\label{defn:landau}
Let $G$ be a Feynman graph. The singularities of $\Phi_G$ form a poset $(\mathcal{S}_G,\leq)$ where
\begin{equation*}
 \mathcal{S}_G:= \{  l_\gamma \mid \gamma \subset G \}
\end{equation*}
and $l_\gamma$ is the stratum of $\mathbb{L}_G$ associated to the subgraph $\gamma \subset G$ as the solution of Landau's equations for
\begin{equation*}
a_e=0 \text{ if } e \in E_{\gamma} \text{ and } Q_e=0 \text{ if } e \in E_{G/\gamma}. 
\end{equation*}
The partial order $\leq$ is given by reverse inclusion,
\begin{equation*}
 l_\gamma \leq l_\eta \Longleftrightarrow \eta \subset \gamma \Longleftrightarrow \mathbb{L}_{G/\gamma} \subset \mathbb{L}_{G/\eta} .
\end{equation*}
The maximal element in this poset is $l_\emptyset$, called the \textit{leading singularity} of $G$. The other elements $l_\gamma$ with $\gamma \neq \emptyset$ are called \textit{non-leading} or \textit{reduced} singularities, the corresponding graphs $G/\gamma$ are referred to as the \textit{reduced graphs} (of $G$). The coatoms in $\mathcal{S}_G$, the elements covered by $l_\emptyset$, are called \textit{next-to-leading} or \textit{almost leading} singularities. In terms of reduced graphs, these coatoms are represented by the graphs $G/e$ where $e \in E_G$.
\end{defn}

\begin{rem}
 Refining this poset structure allows to derive vanishing statements (``Steinmann relations'') for iterated variations/discontinuities of Feynman integrals; see \cite{Marko-Erik}.
\end{rem}

\section{Partial Fractions and Spanning Trees}\label{partialfractions}
In this section we want to derive one of our main results: The computation $\Phi_R(G)$ of a core Feynman graph $G\in H_{core}$ can be obtained as a sum of evaluations of pairs $G_T:=(G,T)$ where $T\in\mathcal{T}(G)$ runs over all spanning trees of $G$ and edges not in the spanning tree are evaluated on-shell. 

We proceed by separating the integration over energy variables $k_0$ for any internal loop momentum $D$-vector $k\in \mathbb{M}^D$ from the space-like integrations for $(D-1)$-vectors  $\vec{k}$.

\subsection{Divided differences}
Consider the integral (see Eq.(\ref{integrandq}))
\[
\int_{-\infty}^\infty \int_{\mathbb{R}^{D-1}}\prod_{i=1}^{|G|}dk_{i;0}d^{D-1}\vec{k}_i I_G^\Pi(\{k_{i;0}\},\{\vec{k}_i\}).
\]
The replacement of $I_G^\Pi$ by $I_G^R$ as in Eq.(\ref{integrandqren}) is understood if renormalization is needed.

We note that any of those $|G|$  energy integrals converges and hence can be done as a residue integral closing the contour say in the upper complex half-plane upon regarding $k_{i;0}$ as a complex variable.

Such multiple residue integrals can be expressed using divided differences \cite{Henrici}.

To this end consider first  a product $\lambda_\gamma$ of $v_\gamma$ quadrics $Q_e$ which constitute a one-loop graph $\gamma$. Without loss of generality we can asumme that each quadric $Q_e$, $e\in E_\gamma$, has the form
\[
Q_e=(k+r_e)^2-m_e^2+i\eta,
\]
for some loop momentum $k$, four-vectors $r_e$, masses $m_e$ and $0\lneq \eta\ll 1$.
We write
\[
\lambda_\gamma:=\prod_{e=1}^{v_\gamma} \frac{1}{Q_e}.
\] 
The divided difference with regard to the function $f: x\to x^{-1}$ 
delivers the partial fraction decomposition 
\be\label{Lpf}
\lambda_\gamma=\sum_{e=1}^{v_\gamma} f(Q_e)\underbrace{\prod_{f \neq e}\frac{1}{Q_f-Q_e}}_{=:\mathbf{pf}_{e}^\gamma}.
\ee
Note that all residues of poles in $\mathbf{pf}_{e}^\gamma$ at $Q_f=Q_e$ vanish by definition of $\mathbf{pf}_{e}^\gamma$.

As an example for the bubble $b$ we find : 
\[
\lambda_b=\frac{1}{Q_1Q_2}=\frac{1}{Q_1}\frac{1}{Q_2-Q_1}+\frac{1}{Q_2}\frac{1}{Q_1-Q_2}=\frac{1}{Q_1-Q_2}\left(f(Q_2)-f(Q_1)\right).
\]

\subsection{$\mathbf{pf}$ and spanning trees}
The edges $f\in E_\gamma$  in $\mathbf{pf}_e^\gamma$ , for $f\not= e$, for any chosen edge $f\in E_\gamma$, form a spanning tree of $\gamma$.

We hence can write
\[
\lambda_\gamma=\sum_{T\in\mathcal{T}(\gamma)}\mathbf{pf}(T)\frac{1}{Q_{\acute{T}}},
\]
where $\acute{T}$ denotes  the edge of $\gamma$ which is not in $T$ so that $\mathbf{pf}(T)=\mathbf{pf}_{\acute{T}}^\gamma$, see Eq.(\ref{Lpf}).

The inverse $\mathbf{pf}(T)^{-1}=(\mathbf{pf}_{\acute{T}}^\gamma)^{-1}$ is linear in the four-vector $k$ and is real for $\eta\not=0$, in particular it vanishes for real values of $k_0$. 

The divided difference structure gives
\begin{prop}
$\lambda_\gamma^{-1}$ does not vanish at any zero of any $\mathbf{pf}(T)^{-1}$.
\end{prop}
\begin{proof}
For $\mathbf{pf}(T)^{-1}$ to vanish, we need to have $\acute{T}$ and $f$ such that $Q_f=Q_{\acute{T}}$. By the divided difference structure the coefficient of this zero is $1/Q_f-1/Q_{\acute{T}}$ which vanishes.
\end{proof}
As a result the poles of $\lambda_\gamma$ in the variable $k_0$ are solely determined by the two zeroes of the quadric $Q_{\acute{T}}$ which are located in the upper and lower complex $k_0$-plane. In particular no spurious infinities arise from poles in $\mathbf{pf}(T)$. Accordingly, upon choosing a dedicated different $0<i\eta_e\ll 1$ for every
quadric $Q_e$ one can show that the resulting residue integrals are independent
of  such a choice \cite{KreimerThesis}. Thus the product of inverse quadrivs regarded as a distribution has a unique definition also when represented as a divided difference.

Indeed,
\[
Q_{\acute{T}}=(k_0+r_{{\acute{T}},0})^2-(\vec{k}+\vec{r}_{\acute{T}})^2-m_{\acute{T}}^2+i\eta,
\]
so that the zeroes are at
\[
{k_0^{\acute{T}}}_\pm=-r_{{\acute{T}},0}\pm \sqrt{(\vec{k}+\vec{r}_{\acute{T}})^2+m_{\acute{T}}^2-i\eta},
\]
and we close the contour in either half-plane with causal boundary conditions as usual.
\subsection{Shifts}
$\lambda_\gamma$ above has to be integrated:
\[
\Phi(\gamma):=\int_{-\infty}^\infty dk_0 \int d^{D-1}\vec{k} \lambda_\gamma.
\]
\begin{prop}
For each term in the partial fraction decomposition the integral
\[
\Phi(\gamma,{T}):=\int_{-\infty}^\infty dk_0 \int d^{D-1}\vec{k} f(Q_{\acute{T}})\mathbf{pf}(T),
\]
exists as a unique Laurent-Taylor  series with a pole of at most first order for 
\[
0<\epsilon\equiv D/2-2\ll 1
\]
and is invariant under the shifts $k_0\to k_0-r_{\acute{T},0}$ and $\vec{k}\to \vec{k}-\vec{r}_{\acute{T}}$.  
\end{prop}
\begin{proof}
Elementary properties of dimensional regularisation \cite{KreimerThesis,Collins}.
\end{proof}
\begin{rem}
A remark on powercounting is in order. Each term in $\Phi(\gamma,T)$ can be more divergent than $\Phi(\gamma)$ itself and only in the sum over spanning trees is the original degree of divergence restored and renormalization achieved through suitable subtractions. That the limit $\epsilon\to 0$ exists for the sum does not imply that it does exist in any summand. Indeed it does not in general and generically $\Phi(\gamma,T)$ is a proper Laurent series with a pole of first order. 
\end{rem}
Assume from now on that for each $\Phi(\gamma,{T})$ the indicated shift has been performed so that $Q_{\acute{T}}=k^2-m^2_{\acute{T}}+i\eta$.
Let 
\[
\mathbf{\bar{pf}}(T)=\mathbf{pf}(T)_{k_0\to k_0-r_{\acute{T},0},
\vec{k}\to \vec{k}-\vec{r}_{\acute{T}}}.
\]
We get
\[
\Phi(\gamma)=\sum_{T\in\mathcal{T}(\gamma)}\Phi(\gamma,{T}):=\int_{-\infty}^\infty dk_0\int  d^{D-1}\vec{k}\sum_{T\in\mathcal{T}(\gamma)}
\mathbf{\bar{pf}}(T)\frac{1}{k_0^2-\vec{k}^2-m_{\acute{T}}^2+i\eta}.
\]
Exchanging the order of integration and doing first the $k_0$-integral by a contour integration closing in the upper complex half-plane we find
for each $T$ 
\[
\Phi(\gamma,{T})=\int d^{D-1}\vec{k}
\mathbf{\bar{pf}}(T)_{|k_0=+\sqrt{\vec{k}^2-m_{\acute{T}}^2+i\eta}}\times\frac{1}{+2\sqrt{\vec{k}^2-m_{\acute{T}}^2+i\eta}}.
\]
Renormalization is understood as needed.

This is of the desired form but has to be generalized to the multi-loop case. 
\subsection{Partial Fractions for generic graphs}
A generalization to multi-loop  graphs proceeds as follows.
We define 
\be\label{Lambda}
\Lambda(Fl_G):=\sum_i \prod_{j=1}^{|G|}\lambda_{\gamma_j^{(i)}}.
\ee
This is a homogeneous polynomial of degree $|G|$ in inverse quadrics $1/Q_e$. The $\gamma_j^{(i)}$ are determined as above in Eq.(\ref{flagg}).

For the unrenormalized integral $\Phi(G)$ on $|G|$ loop momenta $k(j)$, $1\leq j\leq|G|$ we have
\[
\Phi(G):=
\sum_{i=1}^{\xi_G} \left(\prod_{j=1}^{|G|}\int_{\infty}^\infty dk_0(j)\int d^{D-1}\vec{k}(j)\right) \times \left(\prod_{j=1}^{|G|}\lambda_{\gamma_j^{(i)}}\right).
\] 
Note that each flag contributes different residues in the variables $k_0(j)$.

Carrying out all $k_0(j)$-integrals by contour integrations first we find
\be\label{ff}
\Phi(G):=
\sum_{i=1}^{\xi_G} \left(\prod_{j=1}^{|G|}\int d^{D-1}\vec{k}(j)\right) \times \prod_{j=1}^{|G|}
\sum_{T\in\mathcal{T}(\gamma_j^{(i)})}
\mathbf{\bar{pf}}(T)_{k_0(j)=+\sqrt{\vec{k}(j)^2-m_{\acute{T}}^2+i\eta}}
\frac{1}{+2\sqrt{\vec{k}^2-m_{\acute{T}}^2+i\eta}}.
\ee
Note that for each of the $\xi_G$ terms in the above sum, the spannng trees $t_j^i$ of the graphs $\gamma_j^{(i)}$ combine to a spanning tree $T\in\mathcal{T}(G)$. Furthermore each term in the summand indicates one of the $|G|!$ possible orders of the $|G|$ independent cycles of the graph.

As an example let us consider the 3-edge banana graph $\Theta$ of Fig.(\ref{flagtheta}).
We have three quadrics and two loop momenta $k(1)=k,k(2)=l$. The three quadrics are
\bea
Q_1 & = & l_0^2-\vec{l}^2-m_1^2+i\eta\\
Q_2 & = & (l_0-k_0+q_0)^2-\vec{l}^2-\vec{k}^2+2\vec{l}\cdot\vec{k}-m_2^2+i\eta\\
Q_3 & = & k_0^2-\vec{k}^2-m_3^2+i\eta,
\eea
Then, $Q_1$ determines 
\[
l_{0,1}:=+\sqrt{\vec{l}^2+m_1^2},
\]
$Q_2$ determines 
\[
l_{0,2}:=k_0-q_0+\sqrt{\vec{l}^2+\vec{k}^2-2\vec{l}\cdot\vec{k}+m_2^2}
\] for the location of  poles in $l_0$ and
$Q_3$ determines 
\[
k_{0,1}:=+\sqrt{\vec{k}^2+m_3^2},
\] 
while  $Q_2$ determines 
\[
k_{0,2}:=l_0-q_0+\sqrt{\vec{l}^2+\vec{k}^2-2\vec{l}\cdot\vec{k}+m_2^2}
\]
for the location of  poles in $k_0$.

We have
\[
\frac{1}{Q_1Q_2Q_3}=\frac{1}{Q_3}\left(\frac{1}{Q_2-Q_1}\left(\frac{1}{Q_1}-\frac{1}{Q_2}\right)\right)
\]
After an integration of $l_0$, we get
\[
\frac{1}{Q_1Q_2Q_3}\to \frac{1}{Q_3}\left(\frac{1}{Q_1}_{|Q_2=0}\times\frac{1}{l_{0,2}}+\frac{1}{Q_2}_{|Q_1=0}\times\frac{1}{l_{0,1}}\right).
\]
Here $Q_3,Q_{1_{|Q_2=0}},Q_{2_{|Q_1=0}}$ depend on $k_0$.

We can shift $l_0\to l_0+k_0-q_0$ and also for the $k_0$ integration$k_0\to k_0+l_0-q_0$ to obtain the representation in accordance with Eq.(\ref{ff}).

So next we do the $k_0$ integration. $Q_3$ delivers
\[
\frac{1}{Q_1}_{|Q_2=0,Q_3=0}\times\frac{1}{l_{0,2} k_{0,1}}+
\frac{1}{Q_2}_{|Q_1=0,Q_3=0}\times\frac{1}{l_{0,1} k_{0,1}}.
\]
Note that $Q_{1_{|Q_2=0}},Q_{2_{|Q_1=0}}$ have the same zero in $k_0$ by construction
and summing the two terms gives
\[
\frac{1}{Q_3}_{|Q_1=0,Q_2=0}\times\frac{1}{k_{0,2} l_{0,1}}.
\]
Thus integrating the $0$-components delivers a sum of three terms:
\beas
\int_{-\infty}^\infty dk_0dl_0 \frac{1}{Q_1Q_2Q_3} & = & \frac{1}{{Q_2}_{|l_0=l_{0,1},k_0=k_{0,1}}}\times\frac{1}{l_{0,1}k_{0,1}}\\ & + &
\frac{1}{{Q_1}_{|l_0=l_{0,2},k_0=k_{0,1}}}\times\frac{1}{l_{0,2}k_{0,1}}\\ & + & \frac{1}{{Q_3}_{|k_0=k_{0,2},l_0=l_{0,1}}}\times\frac{1}{k_{0,2}l_{0,1}},
\eeas
where we note that after the shifts $k_{0,i}$ and $l_{0,i}$ evaluate to the corresponding $\sqrt{s+m_j^2}$, $s=\vec{k}^2$ or $\vec{l}^2$ in accordance with those shifts.

The sector decomposition $\vec{l}^2>\vec{k}^2$ or
$\vec{l}^2<\vec{k}^2$ then gives the six terms of Fig.(\ref{flagtheta}).
\begin{rem}
The above methods were already used some time ago when using parallel and orthogonal space decompositions in one- and two-loop integrals with masses \cite{KreimerOneLoop,KreimerThesis,KreimerMaster,KreimerTriangle,KreimerBox}.
Integral representations were obtained which facilitated numerical approaches to massive two-loop integrals not available by other methods at the time.

Related methods now emerge under the name of loop-tree duality (LTD). In particular the approach by Hirschi and collaborators \cite{Capatti1,Capatti2,Capatti3,Capatti4} relates to ours
in avoiding spurious singularities at infinity similarly.\footnote{Even when singularities at infinity cancel each other out at the end of the computation they make convergence slow in any approach.} LTD is approached in the literature sometimes slightly differently by modifying the causal structure of propagators, an approach we strictly avoid. See Sec.(II.B.) in \cite{Capatti1} for a critical discussion of such approaches. A valuable discussion of causal structure is given by Tomboulis \cite{Tomboulis}.
\end{rem}
\subsection{General structure}
To understand the structure of this integral it is then useful to count the number of spanning trees of a graph to control its computation.
This is also useful to understand the number of Hodge matrices describing the analytic structure of an evaluated Feynman graph \cite{coaction}.

So we let  $\mathit{spt}(G)=|\mathcal{T}(G)|$ be the number of spanning trees of $G$,
$\mathit{spt}:H_{core}\to \mathbb{N}$, 
and define $\mathbf{spt}:H_{core}\to \mathbb{N}$, $\mathbf{spt}(G):=\mathit{spt}(G)|G|!$.

%

\begin{prop}\label{csptr}
  \mbox{}
  \begin{enumerate}
    \item 
\[
\mathbf{spt}(G)=\sum_{|G^\prime|=1} \mathbf{spt}(G^\prime)\mathbf{spt}(G/G'),
\]
and
\item
If $|G|=1$ and $G$ is bridgeless we have $\mathbf{spt}(G)=spt(G)=e_G$ while for $|G|>1$
\[
\mathbf{spt}(G)=m^{|G|-1}\mathbf{spt}^{|G|}\tilde{\Delta}_{core}^{|G|-1}(G)=m^{|G|-1}\mathit{spt}^{|G|}\tilde{\Delta}_{core}^{|G|-1}(G).
\]
\end{enumerate}
\end{prop}
\begin{proof}(from \cite{Kr-Y})
  \begin{enumerate}
  \item For all $T\in \mathcal{T}(G)$ and $e\in E_G\setminus E_T$ there is a unique cycle in $T\cup e$.  This is called the fundamental cycle $l(T,e)$ associated to $T$ and $e$.  For each spanning tree the fundamental cycles associated to $T$ and each of the edges of $E_G\setminus E_T$ give a basis for the cycle space of $G$.

  Let us count $(T,e)$ pairs with $T\in \mathcal{T}(G)$ and $e\in E_G\setminus E_T$ in two different ways.  Counting directly, there are $\mathit{spt}(G) |G|$ such pairs.  Now we will count $(T,e)$ pairs based on the fundamental cycles.  Each cycle $C$ can appear as a fundamental cycle for any edge $e$ in $C$ and any spanning tree formed from a spanning tree of $G/C$ along with the edges of $C\setminus e$.  So $C$ is the fundamental cycle for $|C|\textit{spt}(G/C)$ $(T,e)$ pairs. So there are $\sum_{|G^\prime|=1} \mathit{spt}(G^\prime)\mathit{spt}(G/G')$ $(T,e)$ pairs in all.  Thus we have
  \[
  \mathit{spt}(G) |G| = \sum_{|G^\prime|=1} \mathit{spt}(G^\prime)\mathit{spt}(G/G').
  \]
  Multiplying both sides by $(|G|-1)!$ gives the result.
  
  \item The $|G|=1$ case is immediate as a bridgeless graph with $|G|=1$ is simply a cycle.  The first equality follows from iterating part $i)$.  To see the same argument directly, note  For any $T\in\mathcal{T}(G)$ the basis of fundamental cycles $\{l_1, \ldots, l_{|G|}\}$ can be ordered in $|G|!$ ways corresponding exactly to the $|G|!$ flags generated by  
$\tilde{\Delta}^{|G|-1}(G)$
\[
l_{i_1}\otimes l_{i_2}/E_{l_{i_1}}\otimes\cdots\otimes l_{i_{|G|}}/(\cup_{j=1}^{|G|-1}E_{l_j}).
\]

Since the $\mathbf{spt}$ on the right of the first equality only acts on one loop graphs it can be replaced by $\mathit{spt}$.
\end{enumerate}
\end{proof}

Note that we can recover $I^\Pi_G$ from each single flag.
\begin{prop}
\[
\Lambda(Fl_G)=\xi_G I^\Pi_G,
\]
\end{prop}
As before $\xi_G$ is the number of distinct flags in $Fl_G$.
\begin{proof}
By definition of $Fl_G$ we can write $Fl_G=\sum_{j=1}^{\xi_G} \gamma_1^{(j)}\otimes\cdots\otimes \gamma^{(j)}_{|G|}$. Each $\lambda(\gamma_k^{(j)})=\prod_{e\in E_{\gamma^{(j)}_k}}\frac{1}{Q_e}$ and we use $\lambda(u\otimes v)=\lambda(u)\lambda(v)$ where we extend $\lambda$ as a map $\lambda:H_{core}\to \mathbb{C}$, $\lambda(G)=\prod_{e\in E_G}\frac{1}{Q_e}$.
\end{proof}
In carrying out all $dk_0(j)$-integrals all $\xi_G$ flags contribute.

Let $\mathit{spt}(G)\equiv |\mathcal{T}_G|$ be the number of spanning trees of $G$.
\begin{lem}
There are $|G|!\times \mathit{spt}(G)=:\mathbf{spt}(G)$ contributing residues. 
\end{lem}
\begin{proof}
Consider a given spanning tree $T\in\mathcal{T}(G)$. The locus $\cap_{e\not\in T} Q_e=0$, 
defines $|G|!$ residues through the $|G|!$ possible orders of evaluation of $\prod_{e\in T}1/Q_e$ corresponding to the $|G|!$ sectors in the above hypercube.

Consider 
\[
\mathrm{Res}_G(T):=\prod_{e\in E_T}\frac{1}{Q_e}.
\]
For any chosen order and fixed chosen $T$, the contour integrals above deliver
\[
\mathrm{Res}_G(T)=\left(\prod_{e\in E_T}\frac{1}{Q_e}\right)_{|l_{0,i}=+\sqrt{s_i+m_i^2}}.
\] 

Next, let us consider the set of residues in the energy integrals which can contribute. Come back to the cycle space $\mathcal{L}_G$ of $G$. Any choice of a spanning tree determines a basis for this space.

Choose an ordering of the cycles $l_i\in\mathcal{L}_G$. This defines a sequence corresponding to some flag
\[
l_1,l_2/l_1,\ldots,l_{|G|}/l_{|G|-1}/\cdots/l_1.
\]

Now any choice of an ordering of the cycles, or equivalently of the edges $e\not\in T$, defines the Feynman integral as an iterated integral, and therefore 
a sequence $s_1>s_2>\cdots>s_{|G|}>0$, where we assign to cycle $l_i$ the variable $s_i=\vec{k(i)}^2$. We get $\mathbf{spt}(G)=spt(G)\times |G|!$
such iterated integrals.
\end{proof}
\subsection{The integral}
Summarising, we have
\begin{thm}\label{phiGT}
The  integral $\Phi_G$ is given as 
\bea
\Phi_G & = & \int_{-\infty}^{\infty}\prod_{i=1}^{|G|}dk_{i,0}
\prod_{j=1}^{|G|}\int d^{D-1}\vec{k}(j)
\frac{1}{\prod_{e\in E_G}Q_e}\nonumber\\
 & = &
\sum_{i=1}^{\xi_G} \left(\prod_{j=1}^{|G|}\int d^{D-1}\vec{k}(j)\right) \times \prod_{j=1}^{|G|}
\sum_{T\in\mathcal{T}(\gamma_j^{(i)})}
\mathbf{\bar{pf}}(T)_{k_0(j)=+\sqrt{\vec{k}(j)^2-m_{\acute{T}}^2+i\eta}}
\frac{1}{+2\sqrt{\vec{k}(j)^2-m_{\acute{T}}^2+i\eta}}.\nonumber
\eea
This can be written as a sum over all spanning trees of $G$ together with a sum of all orderings of the space like integrations in accordance with the flag structure and we find
\[
\Phi_G=\sum_{T\in\mathcal{T}(G)}\Phi_{G_T},
\]
with
\beas
\Phi_{G_T} & = & \sum_{\sigma\in S_{|G|}}\int_{0<s_{\sigma(|G|)}<\cdots<s_{\sigma(1)}<\infty}\left(\prod_{e\in E_T}
\frac{1}{Q_e}\right)_{|k(j)_0^2=s_j+m_j^2,\,j\not\in E_T}\times\\ & &  
\times\frac{s(j)^{\sqrt{D-3}}d^{D-1}\hat{k}(j)}{+2\sqrt{s(j)-m_{\acute{T}}^2+i\eta}}
\prod_{j\not\in E_T}ds(j),
\eeas
where $d^{D-1}\hat{k}(j)$ is the spacelike angular integral over the unit sphere in $(D-1)$ dimensions.
\end{thm}
This is the desired result. If renormalization is needed one 
has to sum over all such terms generated by the corresponding 
Hopf algebra $H_{ren}$ in accordance with the forest formula.

There is a corresponding graphical identity for the energy integrals.
\be\label{graphT}
G=\sum_{T\in\mathcal{T}(G)}G_T,
\ee
where $G_T=(G,T)$ is represented as the cut graph which splits all edges $e\not\in E_T$.
\[
G=(H_G,\mathcal{V}_G,\mathcal{E}_G)\to G_T=(H_G,\mathcal{V}_G,\mathcal{E}_H). 
\]
Here $\mathcal{E}_H$ has as parts of cardinality two only the edges of $T$. 
See Fig.(\ref{treesumcut}).

\begin{rem}
 The preceding theorem has a nice interpretation in terms of parametric Feynman rules viewed as integrals over certain volume forms on moduli spaces of graphs. 
 
 Roughly speaking, these spaces are constructed by taking the integration domains $\mb P_G$ of all Feynman graphs with a fixed number of loops and legs and gluing them together along faces that are indexed by isomorphic graphs. In this picture Feynman rules associate to each cell a family of ``volumes", the values $\Phi_R(G)(p,m)$, parametrized by the kinematic data attached to $G$ (we further discuss this construction in Sec.(\ref{Marko}), for a more detailed exposition see \cite{Marko}).
Each of these moduli spaces deformation retracts onto a certain subspace, its \textit{spine}, which has the structure of a cubical complex, that is, it is a union of cells, each homeomorphic to a (non-degenerate) cube, and the intersection of two cubes is again a cube. In this retraction, each cell of the moduli space, indexed by a graph $G$, gets mapped to a union of cubes, indexed by pairs $(G,F)$ where $F$ is a spanning forest of $G$ \cite{rational,V,MaxMarko}.
Put differently, each cell -- in fact the whole space -- is the total space of a topological fiber bundle with contractible fibers over its spine. When restricted to a singe cell, outside of a subset of measure zero, the bundle map is smooth. See Fig.(\ref{fig:fiberbundle}) for an example. 

Now, integrating along its fibers, we can (at least formally) reduce a parametric Feynman integral $\Phi(G)$ to a sum of parametric integrals over cubes $(G,T)$ where $T$ runs over all spanning trees of $G$.
This is precisely the content of Thm.(\ref{phiGT}), translated from momentum to parametric space.
\end{rem}

\begin{figure}
\begin{tikzpicture}[scale=1]
  \coordinate  (v1) at (0.3,0); 
   \coordinate  (v2) at (2.5,0);
   \coordinate (p1) at (-.2,0); 
   \coordinate (p2) at (3,0); 
   \draw (p1) to (v1);
   \draw (p2) to (v2);
   \draw (v1) to[out=80,in=100] (v2);
   \draw (v1) -- (v2);
   \draw (v1) to[out=-80,in=-100] (v2);
  \filldraw[fill=black] (v1) circle (0.07);
  \filldraw[fill=black] (v2) circle (0.07);
  \end{tikzpicture}
 \begin{tikzpicture}[scale=1]
\coordinate (l) at (-3,0);
  \coordinate  (r) at (3,0); 
   \coordinate  (o) at (0,3);
  \coordinate (lo) at (-1.5,1.5); 
   \coordinate (lr) at (0,0); 
  \coordinate (ro) at (1.5,1.5); 
  \coordinate (c) at (0,1.1); 
  \coordinate (t1) at (-1.2,1.42);
    \coordinate (t2) at (-.8,1.31);
        \coordinate (t3) at (-.3,1.175);
        \coordinate (t4) at (-1,1.36);
        \coordinate (t5) at (-.56,1.24);
   \draw (l) -- (r);
   \draw (l) -- (o) ;
   \draw (r) -- (o) ;
   \draw[blue] (t1) -- (o);
   \draw[blue] (t1) -- (l);
    \draw[blue] (t2) -- (o);
   \draw[blue] (t2) -- (l);
   \draw[blue] (t3) -- (o);
   \draw[blue] (t3) -- (l);
    \draw[blue] (t4) -- (o);
   \draw[blue] (t4) -- (l);
   \draw[blue] (t5) -- (o);
   \draw[blue] (t5) -- (l);
     \draw[red] (lo) -- (c);
   \draw[red] (lr) -- (c);
   \draw[red] (ro) -- (c);
  \filldraw[fill=red] (lo) circle (0.07);
  \filldraw[fill=red] (lr) circle (0.07);
 \filldraw[fill=red] (ro) circle (0.07);
  \filldraw[fill=red] (c) circle (0.07);
 \filldraw[fill=white] (l) circle (0.07);
 \filldraw[fill=white] (r) circle (0.07);
 \filldraw[fill=white] (o) circle (0.07);
  \end{tikzpicture} 
   \caption{A graph $G$ and its corresponding cell $\mb P_G \subset \m {MG}_{2,2}$. The red part is the cubical complex onto which $\mb P_G$ retracts, the blue lines indicate the fibers over one cube $(G,T)$ where $T$ consists of one of the edges of $G$. Note that the three corners of the simplex are not in $\m {MG}_{2,2}$, because collapsing two edges in $G$ alters its loop number.}
	\label{fig:fiberbundle}
\end{figure}
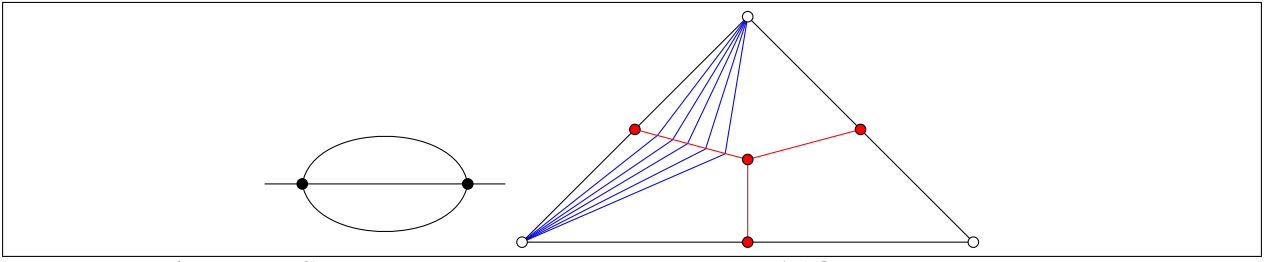

\begin{figure}
\includegraphics[width=12cm]{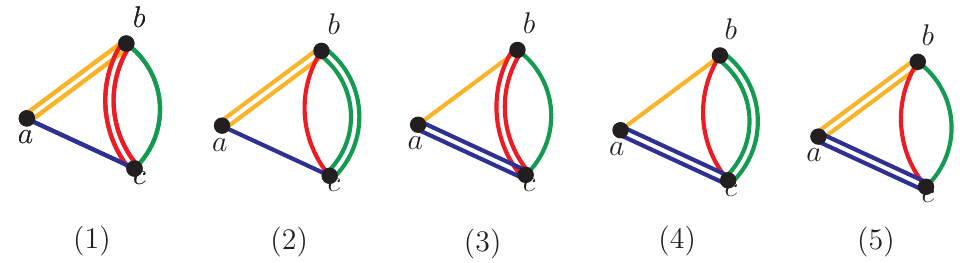}.
\caption{The Dunce's cap $G=dc$ gves rise to five graphs $G_T$, with $T$ running through five spanning trees.
So the five spanning trees give rise to five pairs $G_T$ of graphs and a tree. They are given as a graph with a spanning tree (doubled lines)
where each   edge not in the spanning tree is put on-shell. Hence edges in the spanning tree are off-shell. Edges not in the spanning tree are on-shell.
The external momenta extering at the three vertices are routed through the respective spanning trees. Different colours $(r,b,y,g)$ indicate different masses in each edge.}
\label{treesumcut}
\end{figure}
It is an instructive exercise to the reader to work $\Lambda(Fl_{dc})$ (see Eq.(\ref{Lambda})) and the corresponding identifications out.
We have
\bea
\Lambda(Fl_{dc}) & = & \frac{1}{Q_b}\left(\frac{1}{(Q_r-Q_b)(Q_y-Q_b)}\right)\frac{1}{Q_g}
  +  \frac{1}{Q_y}\left(\frac{1}{(Q_r-Q_y)(Q_b-Q_y)}\right)\frac{1}{Q_g}\nonumber\\
 & + & \frac{1}{Q_r}\left(\frac{1}{(Q_y-Q_r)(Q_b-Q_r)}\right)\frac{1}{Q_g}
  +  \frac{1}{Q_b}\left(\frac{1}{(Q_g-Q_b)(Q_y-Q_b)}\right)\frac{1}{Q_r}\label{dunceexample}\\
 & + & \frac{1}{Q_y}\left(\frac{1}{(Q_g-Q_y)(Q_b-Q_y)}\right)\frac{1}{Q_r}
  +  \frac{1}{Q_g}\left(\frac{1}{(Q_y-Q_g)(Q_b-Q_g)}\right)\frac{1}{Q_r}\nonumber\\
 & + & \frac{1}{Q_b}\left(\frac{1}{(Q_y-Q_b)}\right)\frac{1}{Q_r}\left(\frac{1}{(Q_g-Q_r}\right)
  +  \frac{1}{Q_b}\left(\frac{1}{(Q_y-Q_b)}\right)\frac{1}{Q_g}\left(\frac{1}{(Q_r-Q_g}\right)\nonumber\\
 & + & \frac{1}{Q_y}\left(\frac{1}{(Q_b-Q_y)}\right)\frac{1}{Q_r}\left(\frac{1}{(Q_g-Q_r}\right)
  +  \frac{1}{Q_y}\left(\frac{1}{(Q_b-Q_y)}\right)\frac{1}{Q_g}\left(\frac{1}{(Q_r-Q_g}\right)\nonumber.
\eea
The spanning trees can be read off from Eq.(\ref{dunceexample}) in an obvious way. Each appears twice, for example the spanning tree with edges $e_r,e_y$ contributes to the first and eights term. 

We can also indicate sub- and co-graphs by the edges involved.
Then the first three terms correspond to the contribution
\[
e_be_ye_r\otimes e_g, 
\]
the next three terms correspond to
\[
e_be_ye_g\otimes e_r, 
\]
which gives the $6=3\times 1+3\times 1$ terms of the partial fraction decompositions of the triangle subgraphs and tadpole cographs.
The last four terms give the $4=2\times 2$ terms of the partial fraction decompositions
of the bubble subgraph (on edges $e_r,e_g$) and the bubble co-graph on edges $e_y,e_b$.

The co-graph sub-graph order translates into an order of the spacelike momenta of the loops and hence we find the $10=5\times 2$ terms above as it must be.
This uses that $Res_G(T)$ is uniquely defined for any order of the spacelike momenta.
\begin{rem}
This is all in accordance with a corresponding sector decomposition determined by the spine. For the example of the graph $dc$, see the discussion in \cite{coaction}.
\end{rem}
\section{Cutkosky graphs}
Above, we learned that we should put all edges not in the spanning tree on the mass-shell. Now, for a proper Cutkosky graph $G$, so in the presence of spanning forests instead of spanning trees, we will see that the same message arises: all edges not in the spanning forest will be evaluated on the mass-shell, either due to a contour integral, or due to the fact that they connect distinct components of the forest. 

We are left with only two types of edges:\\
$\bullet$  in the forest
($e\in E_F$),\\
$\bullet$  or not in the forest ($e\in E_{on}=E_{G}-E_F$).
\subsection{The general formula for $H_C^{(0)}$.}
Consider a Cutkosky graph $G\in H_C^{(0)}$ so that no loop is left intact. 

It is generated by a necessarily unique forest $F$ and associated set of edges $E_{on}$ with $e\in E_{on}\Leftrightarrow e\not\in E_F$ so that $E_{on}\dot{\cup} E_F=E_G$.

Then,
\[
\Phi(G)=\int\prod_{j=1}^{|G|}d^Dk_j \left(\frac{1}{\prod_{e\in E_F}Q_e}\right)_{\cap_{f\in E_{on}}(Q_f=0)}.
\]
It remains to describe the threshold divisor prescibed by $\cap_{f\in E_{on}}(Q_f=0)$. 

We first note that $|E_{on}|\gneq |G|$. We can fix more
than the $|G|$ energy variables $k_0(j),\,1\leq j\leq |G|$.
Let us start consider the reduced graph $G_{r}:=G/E_F$
where each edge gives is on-shell and fixes a  variable as this graph is a Cutkosky graph which has all its edges cut.

Any chosen partition of $L_G$ with which  $F$  is compatible defines a partition of $V_{G/E_F}$ and therefore a set of variables $k_{i,0}$ and  $s_i$ which are determined by the set $E_{on}$. As $|E_{on}|\gneq |G|$, all $k_{i,0}$ are fixed and so is at least  $t$, where we set $s_i=t\tilde{s}_i$ for all $i$ and integrate $t$ over the positive real half-axis, whilst the $\tilde{s}_i$ are integrated over a corresponding  simplex $\Delta_\sigma$.

As a result, the $|E_{on}|$ constraints make sure that the remaining integrals are over an integration domain $C_{G/E_F}$ which is compact and give its volume. The computation in Sec.(\ref{trianglecomp}) is a typical example.

Now consider $G$ itself. The side-constraints remain unchanged. The integration domain is still $C_{G/E_F}$ which now splits:
\[
C_{G/E_F}=C_G\times f,
\]
where $f$ is a $e_F$-dimensional fiber such that the integration resulting from the momentum flow through $F$ corresponds to an integral  over this fiber. 
$C_G$ fulfills 
\be\label{trianglecompdiv}
\mathrm{dim}(C_G)=\mathrm{dim}(C_{G/E_f})-e_F.
\ee
Note that the uniqueness of $F$ for a Cutkosky graph in $H_C^{(0)}$ means that we do not have to consider a sum over spanning forests. This is different below when we consider $H_C^{(j)},\,j\gneq 0$.

\begin{rem}
Any 2-partition $V_G=V_1\dot{\cup} V_2$ which is part of a 
$v_G$-refinement of $L_G$ determines a Lorentz scalar \[
s=(\sum_{v\in V_1} q(v))^2
\]
defined from the 2-partition 
$V_G=V_1\dot{\cup} V_2$, the first non-trivial entry in any $v_G$-refinement.
It follows that $\Phi(G)$ has thresholds as a function of $s$ determined by the threshold divisors $\cap_{f\in E_{on}}(Q_f=0)$ with the 2-partition itself providing the normal threshold $s_0$ in that variable $s$. 
\end{rem}
\begin{thm}\label{anomthresholds}
For $G\in H_C^{(0)}$ with $h_0(F)\geq 2$, $\Phi(G)$ exists and determines a threshold $s_F(G)$ in the variable $s$ defined by the 2-partition in a $v_G$-refinement of $L_G$.   
\end{thm}
\begin{proof}
We regard $\Phi_R(G)$ as a function of $s$ only, with all other kinematic variables fixed. The second Symanzik polynomial $\Phi$ is quadratic in edge variables $a_i$ and hence determines a set of discriminants $D(i)$ assigned to such a refinement. Minimizing $s$ under the condition $D(i)=0$ determines the thresholds $s_F(G)$.
\end{proof}
\begin{rem}
The above analysis relies on a decomposition of loop momenta into a parallel space and its orthogonal complement, the former provided by the span of external momenta. 
Hence rephrasing this analysis in terms of Baikov polynomials using the results of \cite{Mastrolia} is an obvious task for future work.   
\end{rem}
\subsection{Using the co-action}\label{usingcoaction}
Let $G$ be a Cutkosky graph with  partition $P$ of $L_G$.

Consider a forest $F$ compatible with $P$ so that we get a pair $G_F$ of a forest $F$ and a graph $G$.
For any such pair there is an associated triple $(G_0,g,F_0)$ where $g\in H_{core}$ and $G_0\in H_C^{(0)}$ so that $|G_0|=0$, which determines $F_0$ uniquely, in accordance with Cor.(\ref{gG0}). The set $\mathcal{F}_P$ of all compatible forests $F$ can be described as
\be\label{FubF}
\mathcal{F}_P=F_0\dot{\cup}\mathcal{T}(g).
\ee
The set $E_{on}^G=E_G-E_F$ so that $E_{on}^{G/g}=E_{G/g}-E_{F_0}$.

Then,
\[
\Phi(G)=\sum_{F\in\mathcal{F}_P}\int\prod_{j=1}^{|G|}d^Dk_j \left(\frac{1}{\prod_{e\in E_F} Q_e}\right)^R_{\cap_{f\in E_{on}^G}(Q_f=0)}.
\]
The superscript ${}^R$ indicated a sum of such terms for renormalization as needed corresponding to the 
transition $I^\Pi_G\to I_G^R$.

Note that this is a variant of Fubini's theorem by Eq.(\ref{FubF}):
\be\label{Fubini}
\Phi_R(G)=\int\prod_{j=1}^{|G/g|}d^Dk_j \left(\frac{1}{\prod_{e\in E_{F_0}}Q_e}\right)_{\cap_{f\in E_{on}^{G/g}}(Q_f=0)}
\underbrace{\sum_{t\in\mathcal{T}(g)}\overbrace{\int\prod_{j=1}^{|g|}d^Dk_j\left(\prod_{e\in E_t}\frac{1}{Q_e}\right)^R_{\cap_{f\in (E_g-E_t)}(Q_f=0)}}^{\Phi_{g_T}}}_{\Phi_R(g)},
\ee
where the superscript $R$ indicates to sum over all terms needed for renormalization as usual, using that the renormalization Hopf algebra $H_{ren}$ is a quotient of $H_{core}$ and co-acts accordingly.

Now consider a $v_G$-refinement $P$ of $L_g$. We call its partitions $P(i)$.
Note that for every $T\in \mathcal{T}(G)$, such a refinement induces an ordering $o_T$ of its edges.

Accompanying the partitions $P(i)$ are Cutkosky graphs $G(i)$, forests $F_0(i)$, reduced graphs $G_r(i)=G(i)/F_i$, core subgraphs $g(i)$, and sets $\mathcal{F}_{P(i)}=F_0(i)\dot{\cup}\mathcal{T}(g(i))$.

With this set-up we thus   get a sequence $\Phi(G(i))$ of evaluations of Cutkosky graphs. They provide the entries in the Hodge matrices studied in \cite{coaction}.

\section{The cubical chain complex, $\Delta_{GF}$ and the pre-Lie product}\label{prelieandccc}
In this section we consider the interplay between 
the Hopf algebra of pairs $(G,F)$ and the boundary $d=d_0+d_1$ of the associated cubical chain complex, and the relation between $d$
and the pre-Lie product
for pairs of Cutkosky graphs $G$ and forests $F$.

This is groundwork to prepare the analysis of cubical complexes in QFT in future work. A first example is studied in the next Sec.(\ref{doneloop}) when analysing the one-loop triangle graph. 

\subsection{The cubical chain complex}\label{ccc}
Consider $G_T\equiv (G,T)$. Any spanning tree $T$ defines  a cube complex for $e_T$-cubes
assigned to $G$. There are $e_T!$ orderings $\mathfrak{o}=\mathfrak{o}(T)$
which we can assign to the edges of $T$.

We define a boundary for any elements $G_F\equiv (G,F)$ of $H_{GF}$. For this consider such an ordering 
\[
\mathfrak{o}:E_F\to
[1,\ldots,e_F]
\]
of the $e_F$ edges of $F$.
There might be other labels assigned to the edges of $G$
and we assume that removing an edge or shrinking an edge will
not alter the labels of the remaining edges. In fact the whole Hopf algebra structure of $H_{core}$ and $H_{GF}$
is preserved for arbitrarily labeled graphs \cite{Turaev}.  

The (cubical) boundary map $d$ is defined by $d:=d_0+d_1$
where
\begin{equation}\label{eq:cubediff}
d_0(G_F^{\mathfrak{o}(F)}):= \sum_{j=1}^{e_F} (-1)^j(G_{F \setdiff e_j}^{\mathfrak{o}(F\setdiff e_j)}), \quad d_1 (G_F^{\mathfrak{o}(F)}) := \sum_{j=1}^{e_F} (-1)^{j-1}({G/e_j}^{\mathfrak{o}(F/e_j)}_{F/e_j}).
\end{equation} 
We understand that all edges $e_k,k\gneq j$ on the right are
relabeled by $e_k\to e_{k-1}$ which defines the corresponding $\mathfrak{o}(T/e_j)$ or $\mathfrak{o}(T \setdiff e_j)$.
Similar if $T$ is replaced by $F$.

From \cite{rational} we know that $d$ is a boundary:
\begin{thm}\cite{rational}
\[
d\circ d=0,\,d_0\circ d_0=0,\, d_1\circ d_1=0.
\]
\end{thm}

Starting from $G_T$ for any chosen $T\in\mathcal{T}_G$
each chosen order $\mathfrak{o}$ defines one of $e_T!$ simplices of the $e_T$-cube. Such simplices will define the triangular matrices studied in \cite{coaction}.

This cubical chain complex was used to compute the homology groups of certain moduli spaces of (Feynman) graphs, see Sec.(\ref{Marko}) and \cite{MaxMarko}.

Here we want to understand how the boundary $d$ interacts with the coproduct $\Delta_{GF}$ and with the pre-Lie
structure which defines $H_{GF}$ dually via the Milnor--Moore--Cartier--Quillen theorem \cite{RHI,MMCQ}. This aims to provide a tool which is hopefully useful in future investigations.
\subsection{The dual $\mathcal{U}(L_{GF})$ of $H_{GF}$}
We have $H_{GF}=\mathcal{U}^\ast(L_{GF})$ by Milnor--Moore--Cartier--Quillen.
The universal enveloping algebra $\mathcal{U}(L_{GF})$ is determined by the Lie algebra $L_{GF}$. The latter comes from a pre-Lie structure which we describe 
below in Sec.(\ref{preLie}).

Here let us recapitulate the set-up.

The Hopf algebra $H_{GF}$ is a commutative $\mathbb{Q}$-algebra and is graded by the loop number. Its linear space of generators $\langle h_{GF}\rangle_{\mathbb{Q}}$ is generated by pairs $(G,F)$ of a graph $G$ and a spanning forest thereof.  By abuse of notation we simply continue to  write $(G,F)$ for a generator 
$h_{(G,F)}\in H_{GF}$ indexed by such a pair.

The boundary $d$ acts as a map 
\[
d:\, H_{GF}\to H_{GF},
\]
by definition.

$H_{GF}=\mathcal{U}^\ast(L_{GF})$ is the dual of a universal enveloping Lie algebra 
$\mathcal{U}(L_{GF})$ and the generators $l_{(G,F)}\in L_{GF}$ of $L_{GF}$ are indexed
by pairs $(G,F)$ themselves. We use the Kronecker pairing
\[
\langle (G,F),l_{(G^\prime,F^\prime)}\rangle=\delta_{G,G^\prime}\delta_{F,F^\prime}.
\]

The Lie bracket in $L_{GF}$
\[
[l_{(G_1,F_1)},l_{(G_2,F_2)}]:=l_{(G_1,F_1)}\ast l_{(G_2,F_2)}-l_{(G_2,F_2)}\ast l_{(G_1,F_1)},
\]
is defined via the pre-Lie product $\ast:\, L_{GF}\otimes L_{GF}\to L_{GF}$. This pre-Lie product $\ast$ itself is determined from the requirement
\beas
\langle \Delta_{GF}(G,F),l_{(G_1,F_1)}\otimes
l_{(G_2,F_2)}-l_{(G_2,F_2)}\otimes l_{(G_1,F_1)} \rangle =
\langle (G,F),
[l_{(G_1,F_1)},l_{(G_2,F_2)}]\rangle.
\eeas 

The so-determined pre-Lie product $\ast$ induces a map (by abuse of notation also denoted by $\ast$)
\beas
\ast: H_{GF}\otimes H_{GF} & \to &  H_{GF},\\
 (G_1,F_1)\ast (G_2,F_2) & = & \sum_i (g^{(i)},f^{(i)})
=\sum_{(g,f)\in \langle H_{GF}\rangle_{\mathbb{Q}}}(g,f)\langle (g,f),l_{(G_1,F_1)\ast (G_2,F_2)}\rangle,
\eeas 
where we sum over all possibilities $(g^{(i)},f^{(i)})$ to insert $(G_2,F_2)$
into $(G_1,F_1)$, see Sec.(\ref{preLie}) below.
Furthermore $l_{(G_1,F_1)\ast (G_2,F_2)}\equiv l_{(G_1,F_1)}\ast l_{(G_2,F_2)} $, $\ast$ is a linear map $L_{GF}\times L_{GF}\to L_{GF}$.

Note that for generators in $L_{GF}$ we use linearity in the subscripts
\[
l_{\sum_i (g^{(i)},f^{(i)})}=\sum_i l_{(g^{(i)},f^{(i)})}.
\]
\subsection{$\Delta_{GF}$ and the boundary $d$}\label{ddeltagfo}
We first investigate the interplay between the map 
$\Delta_{GF}^{\mathfrak{o}}$ defined in Eq.(\ref{HopfPairso})  and the boundary $d$. The fact that a shrunken edge can not be removed  and a removed edge can not shrink allows to treat $d_0$ and $d_1$ individually.

In fact we indicate the action of either boundary on an edge $e$ by marking that edge. We sum over all edges with alternating signs as prescribed by the order $\mathfrak{o}=\mathfrak{o}(F)$ in accordance with Eq.(\ref{eq:cubediff}). 

Similarly for the co-product. We can notate it by coloring edges in $G_F$  with two colors, 'co' (red) and 'sub' (blue) which can be consistently done following 
Sec.(A.4) in \cite{Kr-Y}.

Then applying the coproduct $\Delta_{GF}$  first generates a sum of colored graphs and the boundary map gives a sum of colored graphs where edges $e\in E_F$ are marked (say by a dot) in turn and with signs as prescribed by $\mathfrak{o}(F)$.

Vice versa starting with the boundary $d_0$ or $d_1$ we first mark uncolored edges by a dot
and then color them according to the co-product. The result 
is obviously the same as long as the set of blue edges and the set of red edges are $\mathfrak{o}$-compatible. This is ensured by the definition of $\Delta_{GF}^{\mathfrak{o}}$.

As a result one gets
\[
\Delta_{GF}^{\mathfrak{o}}\circ d=( d\otimes\One+  \iota \otimes d)\circ \Delta^{\mathfrak{o}}_{GF}. 
\]
Here, $\iota$ is the map
\[
\iota: \,(G,F)\to (-1)^{e_F} (G,F),
\]
for $e_F$ the number of edges of $F$. It appears as for an odd number of edges in 
in the term on the lhs of the co-product we get a change of sign in counting.

See Fig.(\ref{cubdD}) for an example.
\begin{figure}[H]
\includegraphics[width=12cm]{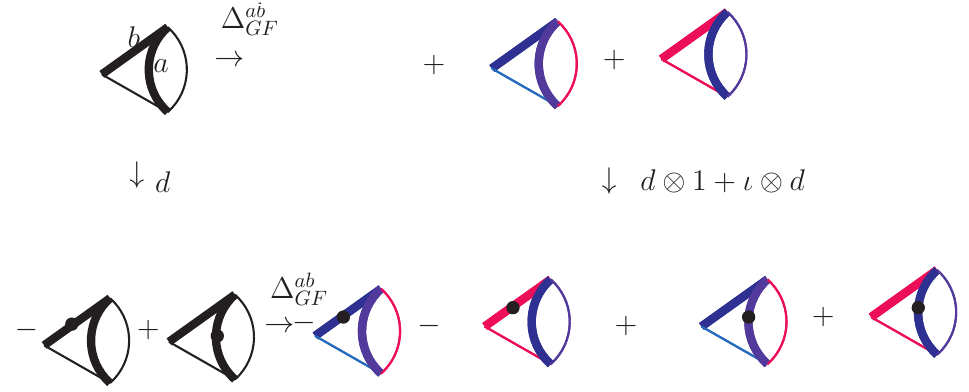}.
\caption{Consider the Dunce's cap on the left with the spanning  tree $T$ indicated by two thickened edges $e_a,e_b$ labeled $a,b$.
The ordered co-product $\Delta_{GF}^{ab}$ can be notated by giving the edges of the subgraphs in blue, and the edges of the co-graphs (obtained by shrinking blue edges) in red. There are two terms generated in the Hopf algebra $H_{GF}$. For the action of $d$ by a dot we indicate the action of either $d_0$ or $d_1$ on the indicated edge of $T$.
With a spanning tree of length two we again get two terms. It follows that we here have $\Delta_{GF}^{ab}d=(d\otimes\One+\iota \otimes d)\Delta_{GF}^{ab}$. Here $d$ can be either $d_0$ or $d_1$.}
\label{cubdD}
\end{figure}
Note that exchanging the order gives the result presented in Fig.(\ref{cubdDinv}).
\begin{figure}[H]
\includegraphics[width=12cm]{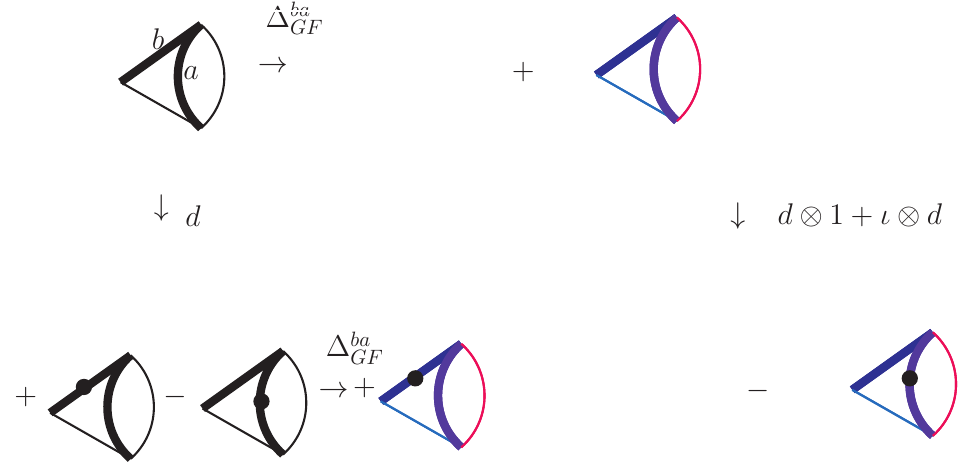}.
\caption{The same for $\Delta_{GF}^{ba}$.}
\label{cubdDinv}
\end{figure}

\subsection{The pre-Lie product for pairs $(G,F)$}\label{preLie}
We define the pre-Lie product $\ast$ as a sum over bijections adopted to pairs $(G,F)$ for $F$ 
a forest. This is well-defined by the Milnor--Moore--Cartier--Quillen theorem. The latter gurantees the existence of a Lie algebra which has an enveloping algebra dual to the Hopf algebra $H_{GF}$. The construction is standard \cite{RHI,Kr-Y},
in particular Sec.(A.4.4) in \cite{Kr-Y} for composing pairs $(G,F)$.

For our purposes we note that when we compose a pair $(G_1,F_1)$ of a graph $G_1$
with an ordered forest $F_1$ with a pair $(G_2,F_2)$ we get a sum
\[
(G_1,F_1)\ast (G_2,F_2)=\sum_i (G_i,F_i)
\]
 of pairs $(G_i,F_i)$
with ordered forests $F_i$.  The orders $\mathfrak{o}(F_i)$ are independent of 
the label $i$, $\mathfrak{o}(F_i)=\mathfrak{o}(F)$ and defined by  concatenation
\[
\mathfrak{o}(F)=\mathfrak{o}(F_2)\mathfrak{o}(F_1).
\]
This is prescribed by Milnor--Moore which imposes that the edges of the sub-graph $F_2$ when inserted are counted before the edges of the co-graph $F_1$.

Finally the sum is over all bijections between external half-edges of $G_2$ with 
suitable half-edges of $G_1$ as described in Sec.(A.4.4) in \cite{Kr-Y}.  
\subsection{Final result}
Let $d=d_0+d_1$ as before, with $d\circ d=0$. As also $d_0\circ d_0=d_1\circ d_1=0$
we have $\{d_0,d_1\}=d_0\circ d_1+d_1\circ d_0=0$.
\begin{thm}\label{dCCC}
i) We can reduce the computation of the boundary map of the cubical chain complex for large graphs to computations for smaller graphs by a Leibniz rule:
\[
d\left( (G_1,F_1)\ast (G_2,F_2)\right)=(d(G_1,F_1))\ast \iota(G_2,F_2)+(G_1,F_1)\ast (d(G_2,F_2)).
\]
ii) We have 
\[
\Delta_{GF}^\mathfrak{o}\circ d=\left(d\otimes\mathrm{id}+\iota\otimes d\right)\circ \Delta_{GF}^\mathfrak{o}. 
\]
\end{thm}
\begin{proof}
All edges of $E_F$ appear either in $E_{F_1}$ or $E_{F_2}$. One by one by  $d$ they either shrink or are removed either in $F_1$ or $F_2$ which makes the signed Leibniz rule in i) obvious, ii) was derived above in Sec.(\ref{ddeltagfo}).
\end{proof}
\section{Monodromy and reduced graphs}\label{doneloop}
We want to use the set-up so far to derive an old result of Polkinghorne et.al.\  \cite{Polkinghorne,ELOP} in the context of one-loop graphs. 
The argument is sufficiently robust to allow for a generalization to the multi-loop case.  
Actually we do a bit more and derive a relation between the amplitude of a reduced graph and the amplitude of the full graph.
\subsection{One-loop graphs}\label{trianglecomp}
Consider the one-loop triangle with vertices $\{A,B,C\}$ and edges $\{(A,B),(B,C),(C,A)\}$,
and quadrics:
$$P_{AB}=k_0^2-k_1^2-k_2^2-k_3^2-M_1,$$ 
$$P_{BC}=(k_0+q_0)^2-k_1^2-k_2^2-k_3^2-M_2,$$ 
$$P_{CA}=(k_0-p_0)^2-(k_1)^2-(k_2)^2-(k_3-p_3)^2-M_3.$$
Here, we Lorentz transformed into the rest frame of the external Lorentz 4-vector $q=(q_0,0,0,0)^T$, and oriented the space like part of $p=(p_0,\vec{p})^T$ in the 3-direction: $\vec{p}=(0,0,p_3)^T$.

Using $q_0=\sqrt{q^2}$, $q_0p_0=q_\mu p^\mu\equiv q.p$, $\vec{p}\cdot\vec{p}=\frac{q.p^2-p.pq.q}{q^2}$, we can express everything in covariant form whenever we want to.

We consider first the two quadrics $P_{AB},P_{BC}$ which intersect in $\mathbb{C}^4$.

The real locus we want to integrate is $\mathbb{R}^4$, and we split this as $\mathbb{R}\times\mathbb{R}^3$,
and the latter three dimensional real space we consider in spherical variables as $\mathbb{R}_+\times S^1\times [-1,1]$,
by going to coordinates  $k_1=\sqrt{s}\sin\phi\sin\theta$,$k_2=\sqrt{s}\cos\phi\sin\theta$,
$k_3=\sqrt{s}\cos\theta$, $s=k_1^2+k_2^2+k_3 ^2$, $z=\cos\theta$.

We have 
$$P_{AB}=k_0^2-s-M_1,$$
$$P_{BC}=(k_0+q_0)^2-s-M_2.$$
So we learn say $s=k_0^2-M_1$ from the first
and $$k_0=k_r:=\frac{M_2-M_1-q_0^2}{2q_0}$$ from the second,
so we set
$$s_r
:=\frac{M_2^2+M_1^2+(q_0^2)^2-2(M_1M_2+q_0^2M_1+q_0^2M_2)}{4q_0^2}.$$

The integral over the real locus transforms to 
$$\int_{\mathbb{R}^4}d^4k\to \frac{1}{2}\int_{\mathbb{R}}\int_{\mathbb{R}_+}\sqrt{s} \delta_+(P_{AB})\delta_+(P_{BC})dk_0ds\times \int_0^{2\pi}\int_{-1}^1 d\phi \delta_+(P_{CA})dz.$$
We consider $k_0,s$ to be base space coordinates, while $P_{CA}$ also depends on the fibre coordinate $z=\cos\theta$. Nothing depends on $\phi$ (for the one-loop box it would).

Integrating in the base and integrating also $\phi$ trivially in the fibre  gives
$$\frac{1}{2} \frac{\sqrt{s_r}}{2q_0}2\pi \int_{-1}^1 \delta_+(P_{CA}(s=s_r,k_0=k_r))dz.$$

For $P_{CA}$ we have 
\be\label{alphabeta}
P_{CA}=(k_r-p_0)^2-s_r -\vec{p}\cdot\vec{p}-2|\vec{p}|\sqrt{s_r}z-M_3=:\alpha+\beta z.
\ee
Integrating the fibre gives a very simple expression (the Jacobian of the $\delta$-function is $1/(2|\sqrt{s_r}\vec{p}|)$, and we are left with
the Omn\`es factor\footnote{For any 4-vector $r$ we have $r^2=r_0^2-\vec{r}\cdot\vec{r}$. Let $q$ be a time-like 4-vector, $p$ an arbitrary 4-vector.
Then, $(q\cdot p^2-q^2 p^2)/q^2=\lambda(q^2,p^2,(q+p)^2)/4q^2$
and in the rest frame of $q$, $(q\cdot p^2-q^2 p^2)/q^2=\vec{p}\cdot \vec{p}$ where $\lambda(a,b,c)=a^2+b^2+c^2-2(ab+bc+ca)$, as always.}
\be\label{Omnes}
\frac{\pi}{4|\vec{p}|q_0}=\frac{\pi}{2|\sqrt{\lambda(q^2,p^2,(q+p)^2)}|}.
\ee

This contributes as long as the fibre variable 
\be\label{fiberz}
z=\frac{(k_r-p_0)^2-s_r -\vec{p}\cdot\vec{p}-M_3}{2\sqrt{\lambda(q^2,p^2,(q+p)^2)s_r}}
\ee
 lies in the range $(-1,1)$.
This is just the condition that the three quadrics intersect.

An anomalous threshold below the normal theshold appears when $(m_1-m_2)^2<q^2<(m_1+m_2)^2$.
In that range, $s_r$ is negative, hence its square root imaginary. 
In the denominator in the expression for $z$ we have the square root of the Kallen function as $|\vec{p}|=|\sqrt{\lambda(q^2,p^2,(p+q)^2)}|/2q_0$. Assume we are not in the rest frame of $q$. 

Then, that Kallen function can be negative as well so that $z$ can still be real. This is then the origin of an anomalous threshold when we solve for the minimal $q^2=q^2(z)$ in the range $1\geq z(q^2)\geq -1$. 

On the other hand, when we leave the propagator $P_{CA}$ uncut,
we have the integral
$$\frac{1}{2} \frac{\sqrt{s_r}}{2q_0}2\pi \int_{-1}^1 \frac{1}{P_{CA}}_{(s=s_r,k_0=k_r)}dz.$$
This delivers a result as foreseen by $S$-Matrix theory \cite{Polkinghorne,ELOP}.

The two $\delta_+$-functions constrain the $k_0$- and $t$-variables, so that the remaining integrals are over the compact domain $S^2$. This is an example alluded to in Eq.(\ref{trianglecompdiv}) where here the fiber is provided by the one-dimensional $z$-integral and the compactum $C_{G/E_F}$ is the two-dimensional $S^2$ while for $C_G$ it is the one-dimensional $S^1$.  

As the integrand does not depend on $\phi$, this gives a result of the form
\be\label{Coverbeta} 
2\pi C \underbrace{\int_{-1}^1 \frac{1}{\alpha+\beta z}dz}_{:=J_{CA}}=2\pi \frac{C}{\beta}\ln\frac{\alpha+\beta}{\alpha-\beta}=\frac{1}{2}\mathrm{Var}(\Phi_R(b_2))\times J_{CA},
\ee
where $C=\sqrt{s_r}/2q_0$ is intimitaly related to $\mathrm{Var}(\Phi_R(b_2))$
for $b_2$ the reduced triangle graph (the bubble), and the factor $1/2$ here is $
\mathrm{Vol}(S^1)/\mathrm{Vol}(S^2)$.

Here, $\alpha$ and $\beta$ are given through (see Eq.(\ref{alphabeta}))
$l_1\equiv \vec{p}^2=\lambda(q^2,p^2,(p+q)^2)/4q^2$ and $l_2:=s_r=\lambda(q^2,M_1,M_2)/4q^2$ as
\[
\alpha=(k_r-p_0)^2-l_2 -l_1-M_3,\,\beta=2\sqrt{l_1l_2}.
\]
Note that
\[
\frac{C}{\beta}=\frac{1}{|\sqrt{\lambda(q^2,p^2,(q+p)^2)}|}=\frac{1}{2q_0|\vec{p}|},
\]
in Eq.(\ref{Coverbeta}) is proportional to the Omn\`es factor Eq.(\ref{Omnes}).

In summary, there is a Landau singularity in the reduced graph in which we shrink $P_{CA}$. It is  located at 
\[
q_0^2=s_{normal}=(\sqrt{M_1}+\sqrt{M_2})^2.
\]
It corresponds to the threshold divisor defined by the intersection $(P_{AB}=0)\cap (P_{BC}=0)$.

 This is not a Landau singularity
when we unshrink $P_{CA}$ though. A (leading) Landau singularity appears
in the triangle when we also intersect the previous divisor with the locus $(P_{CA}=0)$.

It has a location which can be computed from the parametric approach as alluded to in Thm.(\ref{anomthresholds}).   
One finds
\bea\label{sanom} 
q_0^2 & = & s_{anom}  =  (\sqrt{M_1}+\sqrt{M_2})^2+\nonumber\\
& & +\frac{4M_3(\sqrt{\lambda_2}\sqrt{M_1}-\sqrt{\lambda_1}\sqrt{M_2})^2-\left(\sqrt{\lambda_1}(p^2-M_2-M_3)+\sqrt{\lambda_2}((p+q)^2-M_1-M_3)\right)^2}{4M_3\sqrt{\lambda_1}\sqrt{\lambda_2}},\nonumber 
\eea
with $\lambda_1=\lambda(p^2,M_2,M_3)$ and $\lambda_2=\lambda((p+q)^2,M_1,M_3)$. 

Eq.(\ref{Coverbeta}) above is the promised result: the leading singularity of the reduced graph $t/P_{CA}$ and the non-leading singularity of $t$ have the same location and both involve $\mathrm{Var}(\Phi_R(b_2))$ and the non-leading singularity of $t$ factorizes into the (fibre) amplitude $J_{CA}\times \mathrm{Var}(\Phi_R(b_2))$. 

In fact this gives rise to a cycle which is a generator in the above cohomology as Fig.(\ref{cycle}) demonstrates.
\begin{figure}
\includegraphics[width=12cm]{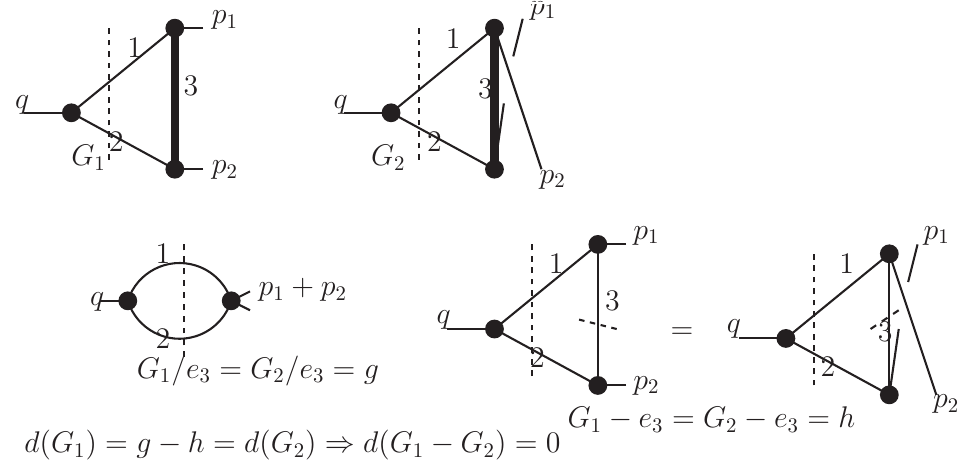}.
\caption{The two Cutkosky  triangle graphs $G_1,G_2$  are distinguished by a permutation of external edges $p_1,p_2$. Edges $e_1,e_2$ are on-shell, $e_3$ is off-shell and hence in the forest. Shrinking or removing it delivers in both cases the same reduced ($g$) or leading ($h$) graph. As a result we get a cycle $d(G_1-G_2)=0$. Obviously there is no $X$ such that $dX=G_1-G_2$.}
\label{cycle}
\end{figure}

To understand how to generalize this it pays to look at the parametric representation.
Consider the second Symanzik polynomial for the triangle graph $\Delta$.
Set $r^2=(p+q)^2$. It reads
\bea
\Phi(\Delta) & = & -M_3A_3^2+A_3(A_1 (p^2-M_1)+A_2 (r^2-M_2))+q^2A_1A_2-(A_1+A_2)(A_1M_1+A_2M_2)\nonumber\\
 & = & \Phi(\Delta/e_3)+A_3\Phi(\Delta-e_3)-A_3^2M_3.\nonumber
\eea
What we are after is the symmetry $r^2\leftrightarrow p^2$ corresponding to the exchange symmetry $p_1\leftrightarrow p_2$ in Fig.(\ref{cycle}). 

For this note that the integration measure is symmetric under the exchange $A_1\leftrightarrow A_2$. As $\Phi(\Delta/e_3)$ has the desired symmetry 
as the two vertices connected by $e_3$ collapse, the result follow from the fact that $\Phi(\Delta-e_3)(A_1,A_2)+\Phi(\Delta-e_3)(A_2,A_1)$ has the desired symmetry. 
\begin{rem}
It is easy to find finite linear combinations of graphs $X=\sum_i G_i$  such that the symmetry of the integration measure enforces $dX=(d_0+d_1)X=0$
similary. The question if there exists $Y$ such that $X=dY$ is harder to answer in general and a systematic study is left to future work. Furthermore factorizations as in Eq.(\ref{Coverbeta}) on the rhs can similarly be established using dispersion relations and will be discussed in future work. 
\end{rem}

\section{Complexes of graphs and Landau singularities}\label{Marko}
Above we have seen how the cubical boundary $d=d_0+d_1$ acts on (cut) Feynman graphs and how it relates to their analytic structure. In this section we study a simpler differential. We forget about the information stored in spanning trees and restrict ourselves to shrinking edges to investigate the role of ``traditional" graph complexes for Feynman graphs and their analytic structure.

By traditional we mean here a differential graded (Hopf) algebra structure on $H_{core}$, induced by the derivation 
\begin{equation*}
    d:H_{core} \longrightarrow H_{core}, \  G \longmapsto \sum_{e=\{v,w\}} \pm G/e,
\end{equation*} defined by collapsing edges that are not tadpoles (cf.\ Defn.(\ref{defn:hg}) below -- the signs $\pm$ are determined by an order on $E_G$; we refrain from a precise definition since in what follows we work exclusively with coefficients in $\mb Z_2$). 

As in Thm.(\ref{dCCC}) we have
\begin{equation}\label{eq:deltad}
 \Delta_{core} (d G) = d^{\otimes} \Delta_{core}(G) 
\end{equation}
where $d^{\otimes}(G' \otimes G''):= dG' \otimes G'' + (-1)^{|G'|} G' \otimes dG''$, using Sweedler's notation. 

To prove this formula for $\mb Z_2$ coefficients\footnote{It holds for integer coefficients as well, but we do not need this for our purposes.}, let $C_G$ denote the set of non-empty core subgraphs $g\subsetneq G$. For any edge $e \in E_G$ we have a decomposition 
\begin{equation*}
 C_{G/e}=  \{ g/e \mid g \in C_G, e\in E_g \} \sqcup \{ g \mid g \in C_G, e\notin E_g \}. 
\end{equation*}
This allows to write the coproduct of $G/e$ as
\begin{align*}
\Delta(G/e)= & \mb I \otimes G/e + G/e \otimes \mb I + \sum_{g \in C_G : e \in E_g} g/e \otimes (G/e)/ (g/e)  + \sum_{g \in C_G: e \notin E_g} g \otimes (G/e)/ g \\
 = &  \mb I \otimes G/e + G/e \otimes \mb I + \sum_{g \in C_G : e \in E_g} g/e \otimes G/ g  + \sum_{g \in C_G: e \notin E_g} g \otimes (G/g)/e.
\end{align*}
If $e$ is a tadpole, then $g/e=0$, by definition, for any $g \subset G$ with $e \in E_g$. The equation above shows thus $\Delta(G/e)=G'/e \otimes G'' + G' \otimes G''/e$ for every $e \in E_G$ and \eqref{eq:deltad} follows.

Apart from this compatibility condition, the map $d$ has another important property: It encodes which Feynman graphs share parts of their Landau varieties.

\subsection{Edge-collapses and the analytic structure of Feynman integrals}

Recall the discussion of Landau singularities of Feynman integrals in Sec.(\ref{Landau}). Given a Feynman graph $G$, the analytic function $\Phi_R(G)$ can in principle be reconstructed by a Hilbert transform from knowledge of its Landau variety $\mb L_G$ and the behavior of $\Phi_R(G)$ in a neighborhood of $\mb L_G$ (the nature of the singularities and the associated monodromy). See \cite{ELOP,Pham} for background material. 

This is a highly non-trivial problem whose solution is not yet fully understood. However, if this reconstruction were indeed possible, we could apply the same method to elements of $H_{core}$, that is, to linear combinations of graphs, or even a full amplitude (say for a fixed number of loops and legs). In this setting it is natural to ask which families of Feynman graphs share a set of singularities -- not only to apply a Hilbert transform directly to linear combinations of graphs, but also to check for possible cancellations of singularities. 

Put differently, one would like to partition the set of graphs contributing to an amplitude into subsets organized by their Landau varieties. 
\newline

\textbf{Disclaimer}: In the following we use the term \emph{singularity} as an abbreviation for the location of a Landau singularity, that is, a solution of the Landau equations. This does not include any classification of the type, or even the prediction whether it is one at all. The Landau equations give only necessary, not sufficient conditions for singularities of Feynman integrals. Here we are only concerned with the Landau variety $\mb L_G$, the set of superficial singularities of $G$, or, more precisely, of the function $\Phi_R(G)$. 
\newline

Considering elements in $H_{core}$, in general each summand in a linear combination of graphs brings its own singularities to the party. However, some graphs will have singularities in common, especially those of non-leading type. Since 
\begin{equation*}
   \mb L_{G+H} \subset \mb L_G \cup \mb L_H
\end{equation*}
holds for all Feynman graphs with the same number of legs, one is naturally led to wonder whether there is an efficient way to partition the set of graphs that contribute to an amplitude so that each subset has ``a large overlap of individual Landau singularities" or ``a small joint Landau singularity".

We argue below that for a theory with cubic interaction this is indeed possible. We construct a partition of the set of graphs contributing to an amplitude into subsets that simultaneously satisfy two properties;
\begin{enumerate}
    \item the integrals, and therefore also their singularities, are related by a symmetry of exchanging masses and/or leg labels,
    \item the Landau singularities have maximal overlap with respect to satisfying such a symmetry.
\end{enumerate}
In formulae: Let $\mb G_{n,s}^m$ denote the set of all Feynman graphs with $n$ loops, $s$ legs and their edges labeled by $m$ different masses. Then (omitting symmetry factors) we can rewrite the Green's function $\m A_{n,s}$ as
\begin{equation*}
    \m A_{n,s}:= \Phi_R \Big( \sum_{G \in \mb G_{n,s}^m} G \Big) =  a_1 + \ldots +a_k := \Phi_R \Big( \sum_{G \in \mb G_1} G \Big) + \ldots + \Phi_R \Big( \sum_{G \in \mb G_k} G \Big)
\end{equation*}
where $\mb G_1 \sqcup \cdots \sqcup \mb G_k$ is a partition of $\mb G_{n,s}^m$. The functions $a_i$ and their singularities (the union of the Landau singularities of the graphs that contribute to them) satisfy the above mentioned maximality (cf.\ Theorem \ref{prop:cycles}) and are related by a $\Gamma$-symmetry where $\Gamma \subset \Sigma_s \times \Sigma_m$ is a subgroup of the group of permutations of legs and mass labels.

The machine that provides this partition is (the top degree homology of) a graph complex whose chains are $\mb Z_2$-linear combinations of generators of $H_{core}$ and the differential is defined by collapsing edges. The connection to Landau singularities is established by the following argument: Via the map $G\mapsto \mb L_G$ each graph represents a subset in the space of external momenta, given by the solution of its Landau equations. The poset structure introduced in Defn.(\ref{defn:landau}) implies that the singularities of all graphs contributing to $\m A_{n,s}$ form a simplicial complex $K=K_{n,s}$. By relating the graph complex to the simplicial chain complex of $K$, we see how (the kernel of) its differential encodes incidence relations of singularities. 

In the next section we present this construction in detail. First, we treat the case of Feynman graphs with all masses different, then we comment on the general case when two or more internal propagators can carry the same non-vanishing mass. 

After that we specialize to the case of a theory with cubic interaction and discuss the above mentioned partition property. We show this to be true for one loop graphs with all masses different by relating the homology of the graph complex to the topology of a certain moduli space of graphs. Roughly speaking, Feynman rules provide a distribution density on this space. Evaluating an amplitude amounts to integrate it against its fundamental class. This class is in fact a sum of cycles (the moduli space is not a manifold) whose elements form the sought-after partition of graphs. 

For the case of general masses we present arguments that support this conjecture. For higher loop numbers the homology of the graph complex is unknown, hindering any further speculations whether such partitions exist in general.

\subsection{Holocolored graphs}

Let us study a toy-model first, Feynman graphs with all edges carrying a different mass. On the graphical level we work thus with graphs where all edges are colored differently, that is, we consider graphs with injective coloring maps $c:E_G \to C$, dubbed \textit{holocolored} graphs. 
If the number of loops $n$ and legs $s$ is fixed, then a simple Euler characteristic argument shows that one needs at least $3(n-1)+s$ colors for each admissible graph to admit such a holocoloring. Here we call a graph \textit{admissible} if it is 1-PI and all of its vertices are of valence at least three.\footnote{Apart from this being the relevant case for physics, this assumption assures the finite-dimensionality of all chain groups and topological spaces we encounter in the following.} 
We write $\mb G_{n,s}$ for the set of all admissible graphs with $n$ loops and $s$ (labeled) legs. For $k\in \mb N$ let $\set k := \{1,\ldots, k\}$.

\begin{defn}\label{defn:hg}
For $n,s \in \mb N$ define a chain complex $(HG,d)=(HG_{n,s},d)$ of holocolored graphs by
\begin{equation*}
HG=HG_{n,s}:= \mb Z_2 \big\langle (G,c) \mid G \in \mb G_{n,s},c:E_G \hookrightarrow \set{3(n-1)+s} \big\rangle,
\end{equation*}
where the grading is given by $|(G,c)|:=|E_G|-1$. The differential $d$ is defined by
\begin{equation*}
d(G,c):= \sum_{ e \in  E_G } (G/e,c_e).
\end{equation*}
Here the coloring $c_e$ is induced by the contraction map, that is, it is simply the restriction of $c$ to $E_G\setdiff \{e\}$. If $e$ is a tadpole, then we set $G/e=0$.
\end{defn}

\begin{rem}
 Many interesting features and applications of graph complexes over fields of characteristic zero stem from the signs in the definition of the differential and their relation to graph automorphisms (see \cite{ConantVogtmann}, for example). Although we do not need the signs here (our graph complexes are thus quite simple), we still have to take automorphisms into account. The automorphism group of a holocolored graph is trivial, but for general colorings these symmetries complicate the picture considerably; see Ex.(\ref{eg:bananas}) and the discussion in the next section. 
\end{rem}

\begin{lem}\label{lem:dsquared}
 $d^2=0$.
\end{lem}

\begin{proof}
 Since we are working over $\mb Z_2$, this is a simple consequence of the fact
\begin{equation*}
 (G/e)/f=G/\{e,f\}=(G/f)/e,
\end{equation*}
which holds for any (colored) graph $G$ and every pair of edges $e,f \in E_G$.
\end{proof}

The differential $d$ maps a graph to the sum of its ``boundary graphs", or in the language of Landau singularities, of its reduced graphs, modulo those obtained by collapsing tadpoles. In terms of the poset of singularities $\mathcal{S}_G$ the image of $d$ is the sum over its coatoms -- cf.\ the discussion at the end of Sec.(\ref{Landau}). Each such coatom represents thereby a family of non-leading singularities of $\Phi_R(G)$ of the form
\begin{equation*}
a_e=0 \text{ and for all } e' \in E_{G/e} \text{ either } a_{e'}=0 \text{ or }Q_{e'}=0.
\end{equation*}
In the poset $\mathcal{S}_G$ these equations correspond to intervals 
\begin{equation*}
 [\emptyset,l_{e}]=\{l \in \mathcal{S}_G \mid l \leq l_{G/e} \}. 
\end{equation*}
Thus, if two graphs $G,H$ satisfy $G/e=H/f$ for some edges $e\in E_G$, $f\in E_H$, the functions $\Phi_R(G)$ and $\Phi_R(F)$ have all corresponding reduced singularities (with $a_e=0$ and $a_f=0$, respectively) in common.

We now state a technical formulation of this condition of shared Landau singularities. Although we will focus later on the case of three-regular graphs (the Feynman diagrams of a theory with only cubic interactions), we state it here in full generality.
The proof uses a geometric picture of the situation where we show how to identify $(HG,d)$ as the simplicial chain complex of a topological space. 
There are two (equivalent) ways of doing so, using
\begin{itemize}
    \item the geometric realization of the simplicial set of Landau singularities,
    \item a moduli space of normalized metrics on Feynman graphs.
\end{itemize}
The former approach was outlined at the end of the previous section, the latter uses the following observation:
Varying the edge-lengths of a graph $G \in \mb G_{n,s}$ parametrizes the interior of the (projective) $|E_G|-1$ dimensional simplex $\mb P_G$ (we mod out overall rescaling). In this regard, parametric Feynman rules can be understood as a map that associates to each Feynman graph $G$ a family of volume forms on the space $\mb P_G$ of (normalized) metrics on $G$.\footnote{Such forms may be ill-defined, i.e., the volume of $\mb P_G$ may be infinite, but it becomes finite after renormalization.} The family is parametrized by the kinematical data - here the external momenta - in the complement of the Landau variety of $G$. Upon integration it produces a multivalued function on the latter space.
The faces of $\mb P_G$ are represented by graphs $H$ obtained from $G$ via sequences of edge-collapses. We define an equivalence relation by declaring two such faces $\mb P_H$ and $\mb P_{H'}$ equivalent if $H$ and $H'$ are isomorphic as colored graphs.
We may thus form a $\Delta$-complex $K=K_{n,s}$ by taking the union of all $\mb P_G$ for $G=(G,c) \in \mb G_{n,s}$ and gluing them together along faces that are equivalent. 

As explained above, this complex gives a geometric picture of the poset of Landau singularities of all graphs in $\mb G_{n,s}$. In this way we see that incidence relations between its simplices capture information about when and where the singularities of their associated Feynman integrals intersect.

\begin{thm}\label{prop:cycles}

Let $X=G_1 + \ldots +G_m  \in HG_{n,s} $ be a cycle of degree $k$. Assume that $X$ is not decomposable as a linear combination of cycles. Then the family $\{G_1, \ldots, G_m \}$ is maximal in the following sense:

\noindent
Write $\mb L_X^{red}$ for the union of all reduced singularities associated to the $G_i$,
\begin{equation*}
\mb L_X^{red} := \bigcup_{i=1}^m  \bigcup_{e \in E_{G_i}} \mb L_{G_i/e}.
\end{equation*}

 \noindent 
 If there is an element $ X'=\sum_{i=1}^{m'}G_i' \in HG_{n,s}$ of degree $k$ with $ \mb L_{X'}^{red} \subseteq  \mb L_X^{red}$ that can be completed to form a different cycle $X' + X'' \in \ker d$, then $ \mb L_{X''}^{red} \nsubseteq  \mb L_X^{red}$.
\end{thm}

It follows that, if $X$ is $d$-closed, then there is no subdivision $X= X_1 + X_2$ (with $dX_i\neq 0$) such that 
\begin{equation*}
 \mb L_{X_1}^{red} \cap \mb L_{X_2}^{red} = \emptyset.
\end{equation*}
Thus, a cycle in $HG$ represents a sum of Feynman integrals, closed under the operation of adding another Feynman integral without generating new (reduced) singularities. Roughly speaking, the graphs in the cycle satisfy two conditions simultaneously; their singularities have maximal overlap while their union is as small as possible. See examples (\ref{eg:holotriangle}) and (\ref{eg:twomasses}) below.

\begin{proof}
There is a natural identification of the elements of $HG_{n,s}$ with the simplicial chains of $K$. Moreover, 
the differential $d$ is \emph{almost} the boundary map of the simplicial chain complex of $K$; the only difference is that in the definition of $d$ we set $G/e=0$ if $e$ is a tadpole. It is therefore a \textit{relative} boundary map in the following sense.
To account for the cancellation of tadpoles, let $I_j$ denote the union of all $j$-dimensional simplices in $K$ that are represented by graphs not in $\mb G_{n,s}$ (i.e., those obtained by collapsing a tadpole in an admissible graph on $j+2$ edges). This allows to identify the homology of $HG_{n,s}$ with a certain relative homology of $K$, 
\begin{equation*}
 H_k(HG_{n,s}) \cong H_k(K,I_{k-1};\mb Z_2) \cong \tilde{H}_k(K/I_{k-1};\mb Z_2).
\end{equation*}

With this geometric interpretation at hand, the theorem now follows from the long exact sequence of a pair. 
Let $Y$ denote the space $K/I_{k-1}$ and, abusing notation, let $X \subset Y$ be the cycle representing the class of $\sum_{i=1}^m G_i$ in $H_k(HG_{n,s})\cong H_k(Y;\mb Z_2)$. The long exact sequence of the pair $X \subset Y$ reads
\begin{equation*}
 \cdots \rightarrow H_k(X;\mb Z_2) \longrightarrow H_k(Y;\mb Z_2) \longrightarrow H_k(Y,X;\mb Z_2) \overset{\partial}{\longrightarrow} H_{k-1}(X;\mb Z_2) \rightarrow \cdots
\end{equation*}
Now, the assumptions on $X'$ imply, under the same abuse of notation, that it represents a class $[X']$ in $H_k(Y,X;\mb Z_2)$. The connecting map $\partial$ maps it to $H_{k-1}(X;\mb Z_2)$, given by the class of the  boundary of $X'$ in $X$. 
If $X'$ is a cycle, then $dX'=0$ and we are done. If it is not a cycle, then $X' \in  \ker \partial$. Since the sequence is exact, there must be an element $X''$ in $H_k(Y;\mb Z_2)$ that gets mapped to $X'$. 
\end{proof}

Note that the reverse implication of Thm.(\ref{prop:cycles}) does not hold. A single graph is in general maximal with respect to the set of its reduced singularities. On the other hand, a full amplitude is always maximal in this sense. This is the reason for our interpretation of cycles in $HG_{n,s}$ as representing the largest possible families with respect to the smallest possible sets of shared singularities. 

The identity $d^2=0$ simply translates into the fact that repeated application of ``reducing'' a graph does not give any new information.
In other words, $d$-exact terms give ``trivial'' relations.

If we specialize to a theory with cubic interaction, then the graphs participating in the $n$-loop $s$-point amplitude are precisely the elements living in degree $3n+s-4$, the highest degree part of $HG$. In this case there are no exact elements so that
\begin{equation*}
 H_{3n+s-4}(HG)= \ker \left( d:(HG_{n,s})_{3n+s-4} \to (HG_{n,s})_{3n+s-5} \right) 
\end{equation*}
detects all cycles. This is the main case we consider in the following.

If we drop the restriction of considering a theory with cubic interaction, then graphs with differing numbers of edges contribute to the amplitude. However, the vertex valency of Feynman graphs in a given theory is usually bounded from above. This restricts the homological degrees that need to be considered to a subset of $\set {3n+s-4}$. 
In this case we would need to take the homology of $HG$ in multiple degrees into account. This has twofold implications: 
\begin{itemize}
 \item Homology detects only closed elements modulo exactness, while with the reasoning we have given here, we are only concerned with the kernel of $d$. Thus, in lower degrees the homology can not predict all relevant elements. However, note that computing $\ker d$ is a simple linear algebra problem. 
 \item Since we are only interested in cycles, we can construct elements with high loop numbers from elements in lower loop numbers (without having to check for exactness), for instance with the maps introduced in Sec.(\ref{sss:higherloops}) below.
\end{itemize}
 In this regard it is also important to note that, although graphs with tadpoles are trivial in kinematic renormalization schemes, we must not omit them in the definition of the graph complexes. They have to be included as ``boundary graphs" to keep track of all reduced singularities of a given graph. 
 
\begin{rem}
\label{rem:loopdegree}
The elements in $H_{3(n-1)+s}(HG)$ may be extended by adding all reduced graphs of each summand, including also graphs of lower loop number that are created by the contraction of subgraphs, as in the definition of the poset structure of $\m S_G$. 
Alternatively, the construction presented here may be adjusted to account for graphs with varying loop numbers. In this case we need to consider \textit{marked weighted graphs} as in \cite{ChanGalatiusPayne} where the term \textit{marking} simply refers to a labeling of the legs while \textit{weights} are additional labels on the vertices which keep track of collapsed loops; see \cite{ChanGalatiusPayne} for a precise definition. This leads to an alternative approach allowing to find classes of Feynman graphs across different loop numbers. The associated graph complex is then related to the topology of a moduli space of \textit{tropical curves}, instead of metric graphs (the latter connection is outlined below).
\end{rem} 

\begin{eg}\label{eg:holotriangle}
 Let us consider the differential of a one loop graph with three legs, 
 \begin{equation}\label{eq:dtriangle}
  \raisebox{-.66cm}{ \begin{tikzpicture}[scale=.8]
   \coordinate  (v1) at (0,0.5); 
   \coordinate  (v2) at (1,0);
   \coordinate (v3) at (1,1);
   \coordinate (p1) at (-.23,0.5); 
   \coordinate (p2) at (1.2,-0.1);
   \coordinate (p3) at (1.2,1.1);
   \draw (p1) to (v1);
   \draw (p2) to (v2);
   \draw (p3) to (v3);
   \draw[red] (v1) -- (v2) node[pos=0.5, below]{$x$};
   \draw[blue] (v2) -- (v3) node[pos=0.5, right]{$y$};
   \draw[teal] (v1) -- (v3) node[pos=0.5, above]{$z$};
  \filldraw[fill=black] (v1) circle (0.04) node[xshift=-.33cm]{$p_1$};
  \filldraw[fill=black] (v2) circle (0.04) node[xshift=.4cm,yshift=-.1cm]{$p_2$};
  \filldraw[fill=black] (v3) circle (0.04) node[xshift=.4cm,yshift=.1cm]{$p_3$};
  \end{tikzpicture}}
 \overset{d}{\longmapsto}
   \raisebox{-.55cm}{ \begin{tikzpicture}[scale=.8]
   \coordinate  (v1) at (0,0.5); 
   \coordinate  (v2) at (1,0.5);
   \coordinate (p1) at (-.3,0.5); 
   \coordinate (p2) at (1.23,0.7);
   \coordinate (p3) at (1.23,0.3);
   \draw (p1) to (v1)node[xshift=-.4cm]{$p_1$} ;
   \draw (p2) to (v2)node[xshift=.4cm,yshift=-.2cm]{$p_2$};
   \draw (p3) to (v2)node[xshift=.4cm,yshift=.2cm]{$p_3$};
   \draw[teal] (v1) to[out=80,in=100]node[pos=0.5, above]{$z$} (v2);
   \draw[red] (v1) to[out=-80,in=-100]node[pos=0.5, below]{$x$} (v2);
  \filldraw[fill=black] (v1) circle (0.04);
  \filldraw[fill=black] (v2) circle (0.04);
  \end{tikzpicture}}
  +
   \raisebox{-.55cm}{ \begin{tikzpicture}[scale=.8]
   \coordinate  (v1) at (0,0.5); 
   \coordinate  (v2) at (1,0.5);
   \coordinate (p1) at (-.3,0.5); 
   \coordinate (p2) at (1.23,0.7);
   \coordinate (p3) at (1.23,0.3);
   \draw (p1) to (v1)node[xshift=-.4cm]{$p_2$} ;
   \draw (p2) to (v2)node[xshift=.4cm,yshift=-.2cm]{$p_3$};
   \draw (p3) to (v2)node[xshift=.4cm,yshift=.2cm]{$p_1$};
   \draw[blue] (v1) to[out=80,in=100]node[pos=0.5, above]{$y$} (v2);
   \draw[red] (v1) to[out=-80,in=-100]node[pos=0.5, below]{$x$} (v2);
  \filldraw[fill=black] (v1) circle (0.04);
  \filldraw[fill=black] (v2) circle (0.04);
  \end{tikzpicture}}
  +
   \raisebox{-.55cm}{ \begin{tikzpicture}[scale=.8]
   \coordinate  (v1) at (0,0.5); 
   \coordinate  (v2) at (1,0.5);
   \coordinate (p1) at (-.3,0.5); 
   \coordinate (p2) at (1.23,0.7);
   \coordinate (p3) at (1.23,0.3);
   \draw (p1) to (v1)node[xshift=-.4cm]{$p_3$} ;
   \draw (p2) to (v2)node[xshift=.4cm,yshift=-.2cm]{$p_1$};
   \draw (p3) to (v2)node[xshift=.4cm,yshift=.2cm]{$p_2$};
   \draw[blue] (v1) to[out=80,in=100]node[pos=0.5, above]{$y$} (v2);
   \draw[teal] (v1) to[out=-80,in=-100]node[pos=0.5, below]{$z$} (v2);
  \filldraw[fill=black] (v1) circle (0.04);
  \filldraw[fill=black] (v2) circle (0.04);
  \end{tikzpicture}}.
 \end{equation}
 From this it readily follows that the sum over all six permutations of colorings by $x,y,z$ defines a cycle, hence a generator of $H_2(HG_{1,3})$. There are no other graphs in $HG_{1,3}$ with three edges, so $H_2(HG_{1,3})\cong \mb Z_2$ -- in accordance with \eqref{eq:tophom} below. 
 On the level of Landau singularities we find for the graph on the left hand side of \eqref{eq:dtriangle} reduced singularities for $p_1^2=(x\pm z)^2$, $p_2^2=(x\pm y)^2$ and $p_3^2=(y\pm z)^2$. From this it is also clear, that $\Phi$ applied to any two graphs that are related by a permutation of $x,y,z$ produces two functions which have some singularities in common.
  \end{eg}
  
  \begin{eg}
   For the case of four legs we would find three different generators, each given by the sum of 24 box graphs, their set of reduced singularities related by an $\Sigma_3$-symmetry. Instead of presenting the full computation, we refer to the general discussion below.
  \end{eg}

\subsubsection{The homology of $HG_{1,s}$}

In the case of one loop holocolored graphs with $s$ legs the top degree homology of $HG=HG_{1,s}$ was computed in \cite{MaxMarko}. It is given by the formula
\begin{equation} \label{eq:tophom}
 H_{s-1}(HG_{1,s}) \cong \mb Z_2^{\frac{(s-1)!}{2}},
\end{equation}
In this section we provide a geometric way of understanding these homology groups.
\newline

For $n=1$ the complexes $HG_{1,s}$ are naturally isomorphic to the simplicial chain complexes of certain $\Delta$-complexes, constructed as in the proof of Thm.(\ref{prop:cycles}): Take the union of all $\mb P_G$ for $G=(G,c)\in \mb G_{n,s}$ and glue them together along faces that correspond to isomorphic colored graphs (c.f.\ \cite{Marko,MaxMarko} for details).
Since in the one loop case there are no tadpoles to collapse, every edge-collapse represents such a face relation, and vice versa.
The disjoint union of all simplices $\mb P_G$ associated to holocolored graphs in $\mb G_{1,s}$, glued together via the above described face relations, forms thus a pure\footnote{A $\Delta$-complex of dimension $d$ is \textit{pure} if every simplex is the face of a $(d+1)$-simplex.} $\Delta$-complex of dimension $s-1$, the \textit{moduli space of holocolored one loop graphs with $s$ legs} $\m {MHG}_{1,s}$.

Clearly, there is a one-to-one correspondence between the simplices in $\m {MHG}_{1,s}$ and the elements of $HG_{1,s}$ under which the map $d$ transforms into the simplicial boundary map. This induces a chain isomorphism 
\begin{equation*}
 HG_{1,s} \overset{\cong}{\longrightarrow} C_*(\m {MHG}_{1,s};\mb Z_2),
\end{equation*}
so that
\begin{equation*}
 H_*(HG_{1,s})  \cong H_*(\m {MHG}_{1,s};\mb Z_2).
\end{equation*}
Moreover, if we define orientations on graphs by ordering their internal edges, this isomorphism extends to the case of integer coefficients \cite{MaxMarko}.

The top-dimensional facets of $\m {MHG}_{1,s}$ may be represented by cyclic graphs with $s$ labeled vertices/legs and $s$ colors on their internal edges. Traveling from one such facet to its neighbor is in this representation expressed by exchanging two neighboring legs while keeping the same color pattern on the edges. We call this operation a \textit{leg-flip}. See Fig.(\ref{fig:cyclerep}) for an example.
\begin{figure}[h!]
\hspace{0.8cm}
\begin{tikzpicture}[scale=1]
\coordinate (v0) at (0,-.8);
\draw[fill=white] (v0) circle;
  \coordinate  (v1) at (0,0); 
   \coordinate  (v2) at (1.5,0);
   \coordinate (v3) at (1.5,1); 
   \coordinate (v4) at (0,1); 
  \coordinate (p1) at (-0.5,0); 
   \coordinate (p2) at (2,0); 
   \coordinate (p3) at (2,1); 
   \coordinate (p4) at (-0.5,1); 
   \draw (p1) -- (v1) node[pos=0.1,left]{$p_2$};
   \draw (p2) -- (v2) node[pos=0.1,right]{$p_3$} node[pos=0,xshift=0.9cm, yshift=0.5cm]{$\leftrightarrow$};
   \draw (p3) -- (v3) node[pos=0.1,right]{$p_4$};
   \draw (p4) -- (v4) node[pos=0.1,left]{$p_1$};
       \draw[red] (v1) -- (v2) ;
       \draw[blue] (v2) -- (v3) ;
       \draw[teal] (v3) -- (v4) ;
       \draw[cyan] (v4) -- (v1);
  \filldraw[fill=black] (v1) circle (0.05);
  \filldraw[fill=black] (v2) circle (0.05);
  \filldraw[fill=black] (v3) circle (0.05);
  \filldraw[fill=black] (v4) circle (0.05);
  \end{tikzpicture}  
  \begin{tikzpicture}[scale=1]
\coordinate (v0) at (0,-.8);
\draw[fill=white] (v0) circle;
  \coordinate  (v1) at (0,0.5); 
   \coordinate  (v2) at (1.5,0);
   \coordinate (v3) at (1.5,1); 
   \coordinate (v4) at (0,0.5); 
  \coordinate (p1) at (-0.5,0); 
   \coordinate (p2) at (2,0); 
   \coordinate (p3) at (2,1); 
   \coordinate (p4) at (-0.5,1); 
   \draw (p1) -- (v1) node[pos=0.1,left]{$p_2$};
   \draw (p2) -- (v2) node[pos=0.1,right]{$p_3$} node[pos=0,xshift=0.9cm, yshift=0.5cm]{$\leftrightarrow$};
   \draw (p3) -- (v3) node[pos=0.1,right]{$p_4$};
   \draw (p4) -- (v4) node[pos=0.1,left]{$p_1$};
       \draw[red] (v1) -- (v2) ;
       \draw[blue] (v2) -- (v3) ;
       \draw[teal] (v3) -- (v4) ;
       \draw[cyan] (v4) -- (v1);
  \filldraw[fill=black] (v1) circle (0.05);
  \filldraw[fill=black] (v2) circle (0.05);
  \filldraw[fill=black] (v3) circle (0.05);
  \filldraw[fill=black] (v4) circle (0.05);
  \end{tikzpicture} 
 \begin{tikzpicture}
 \coordinate (v0) at (0,-.8);
\draw[fill=white] (v0) circle;
  \coordinate  (v1) at (0,0); 
   \coordinate  (v2) at (1.5,0);
   \coordinate (v3) at (1.5,1); 
   \coordinate (v4) at (0,1); 
  \coordinate (p1) at (-0.5,0); 
   \coordinate (p2) at (2,0); 
   \coordinate (p3) at (2,1); 
   \coordinate (p4) at (-0.5,1); 
   \draw (p1) -- (v1) node[pos=0.1,left]{$p_2$};
   \draw (p2) -- (v2) node[pos=0.1,right]{$p_3$};
   \draw (p3) -- (v3) node[pos=0.1,right]{$p_4$};
   \draw (p4) -- (v4) node[pos=0.1,left]{$p_1$};
       \draw[red] (v4) -- (v2) ;
       \draw[blue] (v2) -- (v3) ;
       \draw[teal] (v1) -- (v3) ;
       \draw[cyan] (v4) -- (v1);
  \filldraw[fill=black] (v1) circle (0.05);
  \filldraw[fill=black] (v2) circle (0.05);
  \filldraw[fill=black] (v3) circle (0.05);
  \filldraw[fill=black] (v4) circle (0.05);
  \end{tikzpicture}  
  \newline
    \begin{tikzpicture}[scale=0.7]
  \coordinate  (v1) at (1,0); 
  \coordinate  (p1) at (1.5,0); 
     \coordinate  (v2) at (0,1);
     \coordinate  (p2) at (0,1.5);
       \coordinate (v3) at (-1,0); 
       \coordinate (p3) at (-1.5,0);
          \coordinate (v4) at (0,-1);   
          \coordinate (p4) at (0,-1.5);
   \draw[red] (v1) arc[radius= 1cm, start angle= 0, end angle= 90];
   \draw[blue] (v2) arc[radius= 1cm, start angle= 90, end angle= 180];
   \draw[teal] (v3) arc[radius= 1cm, start angle= 180, end angle= 270];
   \draw[cyan] (v4) arc[radius= 1cm, start angle= 270, end angle= 360];
   \draw (p1) -- (v1) node[pos=0.1,right]{$p_2$} node[pos=0,xshift=1.05cm, yshift=0.1cm]{$\overset{\sigma_1}{\longleftrightarrow}$};
   \draw (p2) -- (v2) node[pos=0.1,right]{$p_3$};
   \draw (p3) -- (v3) node[pos=0.1,left]{$p_4$};
   \draw (p4) -- (v4) node[pos=0.1,right]{$p_1$};
   \filldraw[fill=gray] (v1) circle (0.05);
  \filldraw[fill=gray] (v2) circle (0.05);
  \filldraw[fill=gray] (v3) circle (0.05);
  \filldraw[fill=gray] (v4) circle (0.05);
  \end{tikzpicture} 
  \begin{tikzpicture}[scale=0.7]
    \coordinate  (v1) at (1,0); 
  \coordinate  (p1) at (1.5,0); 
     \coordinate  (v2) at (0,1);
     \coordinate  (p2) at (0,1.5);
       \coordinate (v3) at (-1,0); 
       \coordinate (p3) at (-1.5,0);
          \coordinate (v4) at (0,-1);   
          \coordinate (p4) at (0,-1.5);
   \draw[red] (v1) arc[radius= 1cm, start angle= 0, end angle= 90];
   \draw[blue] (v2) arc[radius= 1cm, start angle= 90, end angle= 180];
   \draw[teal] (v3) arc[radius= 1cm, start angle= 180, end angle= 270];
   \draw[cyan] (v4) arc[radius= 1cm, start angle= 270, end angle= 360];
   \draw (p1) -- (v1) node[pos=0.1,right]{$p_1$};
   \draw (p2) -- (v2) node[pos=0.1,right]{$p_3$};
   \draw (p3) -- (v3) node[pos=0.1,left]{$p_4$};
   \draw (p4) -- (v4) node[pos=0.1,right]{$p_2$};
   \filldraw[fill=gray] (v1) circle (0.05);
  \filldraw[fill=gray] (v2) circle (0.05);
  \filldraw[fill=gray] (v3) circle (0.05);
  \filldraw[fill=gray] (v4) circle (0.05);
  \end{tikzpicture} 
   \caption{Two graphs representing two neighboring facets of $\m {MHG}_{1,4}$ and their representatives, related by a leg-flip $\sigma_1$, interchanging the legs carrying momenta $p_1$ and $p_2$. In geometric terms, we travel in $\m {MHG}_{1,4}$ from one facet to the other through the codimension one face represented by the graph obtained from the two in the figure by collapsing the cyan colored edge.}
	\label{fig:cyclerep}
\end{figure}
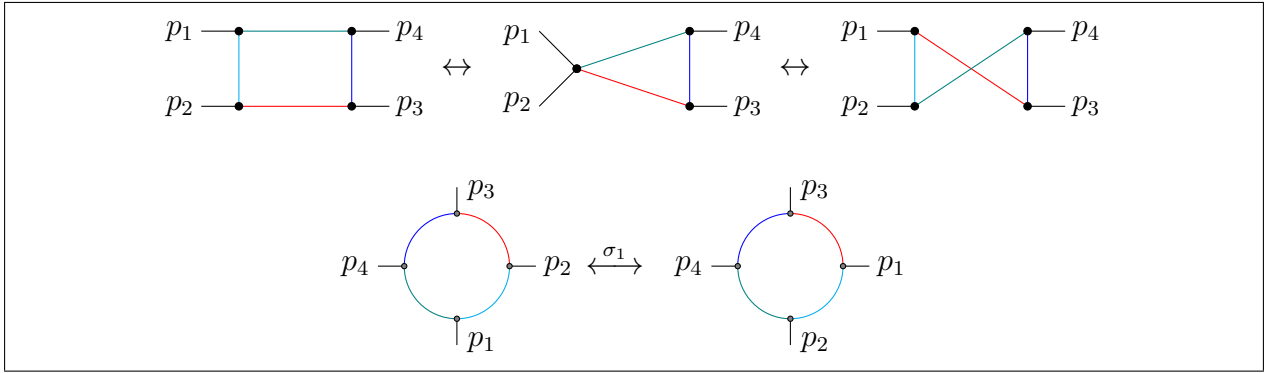
In the one loop case every permutation of legs can be expressed as a sequence of leg-flips. This generates a free $\Sigma_s$-action on $\m {MHG}_{1,s}$.

\begin{prop}
The action of $\Sigma_s$ on (the top-dimensional facets of) $\m {MHG}_{1,s}$ is free with $\frac{1}{2}(s-1)!$ different orbits.
\end{prop}
\begin{proof}
We use the cyclic representation introduced above. A cycle graph $C_s$ on $s$ vertices has the dihedral group $D_s$ as its group of automorphisms. Since $|D_n|=2s$ and there are $s!$ possible colorings of its edges, we have $\frac{1}{2}(s-1)!$ non-isomorphic colorings. 

Take any such coloring $c$ and consider the colored graph $(C_s,c)$. In addition to the coloring of its edges, the graph has $s$ labeled legs attached to it, which is equivalent to an order on its $s$ vertices. Thus, every edge and every vertex of $(C_s,c)$ is uniquely labeled, so this graph cannot have any automorphisms. In particular, for two non-isomorphic choices of colorings, there is no permutation of its vertices that translates one into the other. Hence, the action is free, and its set of coinvariants consists of the $\frac{1}{2}(s-1)!$ non-isomorphic colorings of $C_s$. 
\end{proof}

These orbits are full $(s-1)$-dimensional subcomplexes of $\m {MHG}_{1,s}$ that intersect each other only in faces of codimension greater than two. Thus, for calculating homology in dimension $s-1$ it suffices to consider each subcomplex individually.
Eq.\eqref{eq:tophom} follows now from the simple observation that in each subcomplex each $(s-2)$-dimensional simplex appears as a codimension one face of exactly two top-dimensional facets, related by a leg-flip. Therefore, the sum over all elements of a $\Sigma_s$-orbit represents a homology class.\footnote{It may be interpreted as the fundamental class of the non-singular part of $\m {MHG}_{1,s}$ that is covered by this orbit.} Moreover, all classes arise in such manner.

\subsubsection{Digression: Homology with integer coefficients}

The result holds also for homology with integer coefficients, that is, there are no torsion elements in $H_*(HG_{1,s};\mb Z) \cong H_*(\m {MHG}_{1,s};\mb Z)$.
To see this we need to introduce the notion of a two-coloring of a $\Delta$-complex.

\begin{defn}
 Let $K$ be a $\Delta$-complex. A \textit{two-coloring} of $K$ is an assignment of labels in $\{+,-\}$ to each of its top-dimensional facets, such that no two facets that are both labeled by $+$ or $-$, share a codimension one face.
 A $\Delta$-complex $K$ is called \textit{two-colorable} if it admits a two-coloring.
\end{defn}

We will deduce \eqref{eq:tophom} with integral coefficients by showing that the complexes $\m {MHG}_{1,s}$ are two-colorable. This, together with the above result for $\mb Z_2$-coefficients, implies that we can orient each simplex in a $\Sigma_s$-orbit in such a way that the (oriented) boundary of the sum of its (oriented) elements vanishes.

By the same reasoning as above, to find a two-coloring of the total complex $\m {MHG}_{1,s}$, it suffices to consider each of its $\frac{1}{2}(s-1)!$ $\Sigma_s$-invariant subcomplexes. For this let us look at the dual graphs of these subcomplexes. 
Here, the \textit{dual graph} of a pure $\Delta$-complex $K$ is the graph $G_K$ defined by
\begin{align*}
V(G_K) & := \{ \Delta \mid \Delta \text{ is a top-dimensional facet of $K$} \}, \\
E(G_K) & := \{ (\Delta,\Delta') \mid \Delta \cap \Delta' \text{ is a codimension one face} \}. 
\end{align*}

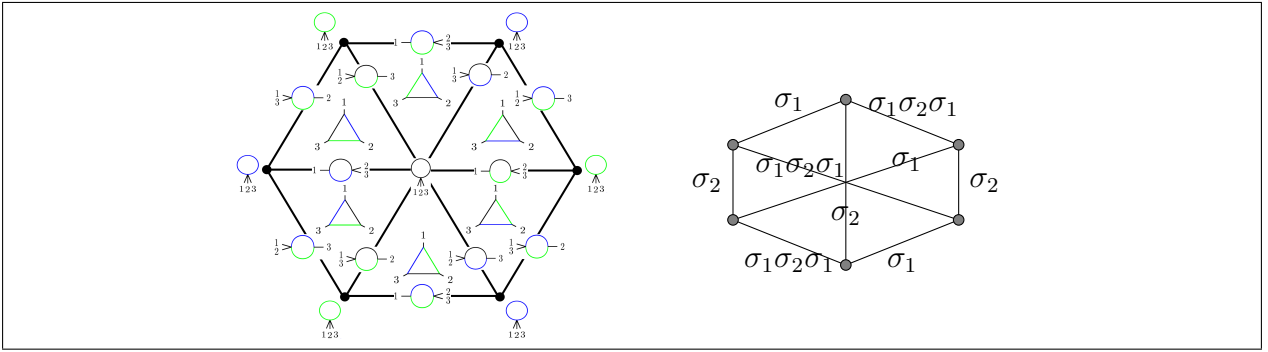
\begin{figure}[t]
	\includegraphics[width=5cm]{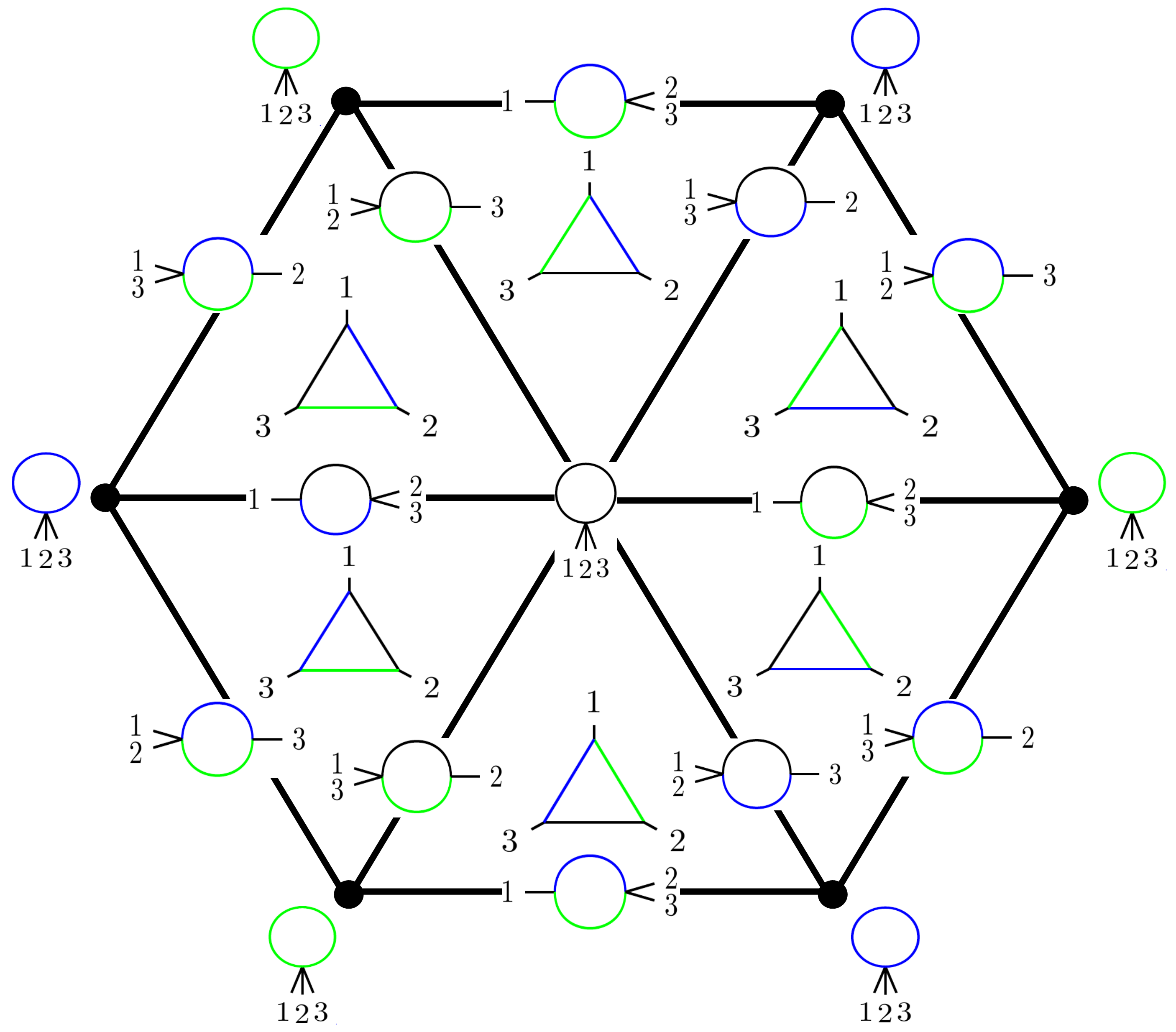}
\begin{tikzpicture}[scale=1]
\coordinate (v0) at (-1.5,-2.6);
 \filldraw[fill=white] (v0) circle;
  \coordinate  (v1) at (0,0); 
   \coordinate  (v2) at (1.5,.6);
   \coordinate (v3) at (3,0); 
   \coordinate (v4) at (3,-1); 
   \coordinate (v5) at (1.5,-1.6); 
   \coordinate (v6) at (0,-1); 
   \draw (v1) -- (v2) node[pos=0.5,above]{$\sigma_1$};
   \draw (v2) -- (v3) node[pos=0.6,above]{$\sigma_1 \sigma_2 \sigma_1$};
   \draw (v3) -- (v4) node[pos=0.5,right]{$\sigma_2$};
   \draw (v4) -- (v5) node[pos=0.5,below]{$\sigma_1$};
   \draw (v5) -- (v6) node[pos=0.5,below]{$\sigma_1 \sigma_2 \sigma_1$};
   \draw (v6) -- (v1) node[pos=0.5,left]{$\sigma_2$};
       \draw (v1) -- (v4) node[pos=0.3]{$\sigma_1 \sigma_2 \sigma_1$};
       \draw (v2) -- (v5) node[pos=0.7]{$\sigma_2$};
       \draw (v3) -- (v6) node[pos=0.23]{$\sigma_1$};
  \filldraw[fill=gray] (v1) circle (0.07);
  \filldraw[fill=gray] (v2) circle (0.07);
  \filldraw[fill=gray] (v3) circle (0.07);
  \filldraw[fill=gray] (v4) circle (0.07);
  \filldraw[fill=gray] (v5) circle (0.07);
  \filldraw[fill=gray] (v6) circle (0.07);
  \end{tikzpicture} 
   \caption{The $\Delta$-complex $\m {MHG}_{1,3}$ and its dual graph with edges labeled by the corresponding leg-flips ($\sigma_i$ flips legs $i$ and $i+1$).}
	\label{fig:s3r}
\end{figure}

In the present case, the dual graph of a $\Sigma_s$-orbit can be described as follows:
Its vertices are given by cyclic graphs with $s$ edges and $s$ legs, the edges colored by a fixed color pattern (there are $\frac{1}{2}(s-1)!$ non-isomorphic choices, corresponding to each orbit/subcomplex), the legs labeled by elements in $\{1, \ldots, s\}$.
Two such vertices are adjacent if and only if the corresponding cyclic graphs are related by a leg-flip. It is therefore a simple graph. The integral version of the formula in Eq.\eqref{eq:tophom} now follows from

\begin{thm}\label{thm:two-colors}
For all $s\geq 1$ the $\Delta$-complex $\m {MHG}_{1,s}$ is two-colorable. 
\end{thm}

The proof relies on two propositions on the colorability of graphs, which we apply to the dual graphs of the $\Sigma_s$-invariant subcomplexes of $\m {MHG}_{1,s}$. For a definition of the graph-theoretic notions and proofs of the following two statements, see \cite{BondiMurty}.

\begin{prop}\label{prop:graphcolor}
Let $G$ be a finite simple graph. $G$ is two-colorable if and only if it is bipartite. 
\end{prop}

\begin{prop}\label{prop:graphbipartite}
Let $G$ be a finite simple graph. $G$ is bipartite if and only if it contains no odd cycles.
\end{prop}

\begin{proof}[Proof of Theorem \ref{thm:two-colors}]
Let $G$ be the dual graph of one of the $\Sigma_s$-invariant subcomplexes of $\m {MHG}_{1,s}$, determined by fixing a color pattern. Since every vertex corresponds to a leg configuration and the adjacency relation in $G$ is given by leg-flips, we have an induced $\Sigma_s$-action on $G$. Therefore, cycles in $G$ are in one-to-one correspondence with closed orbits of the $\Sigma_s$-action. 

Since this action is free (it is induced by the free action of $\Sigma_s$ on $\m {MHG}_{1,s}$), the only way to form a cycle is by a relation in the presentation of $\Sigma_s$ with leg-flips.
 Using the well-known fact that
\begin{equation*}
\Sigma_s = \langle \sigma_1, \ldots, \sigma_{s-1} \mid \sigma_i^2=e, \ \sigma_i\sigma_{i+1} \sigma_i= \sigma_{i+1} \sigma_i \sigma_{i+1} , \ \sigma_i \sigma_j = \sigma_j \sigma_i \text{ for } |i-j|>1 \rangle,
\end{equation*}
we deduce that the only possible cycles in $G$ are of length six (the cycles of length two are trivial). Applying Prop.(\ref{prop:graphcolor}) and Prop.(\ref{prop:graphbipartite}) finishes the proof.
\end{proof}

For the lower degree homology groups of $\m {MHG}_{1,s}$ partial results exist from computer calculations. These, together with lists of generators can be found in \cite{MaxMaster}. 

Recall that in Sec.(\ref{doneloop}) we gave an explicit formula for the variation associated to singularities of one loop graphs, cf.\ Eq.\eqref{Coverbeta}. It shows that graphs with a common boundary term share the location of reduced singularities and the coresponding variation; the latter is expressed by a single function with only its constant coefficients depending on the distribution of colors/masses on the graph.

\subsubsection{Higher loop numbers}\label{sss:higherloops}

For higher loop numbers the homology of $HG_{n,s}$ is not known. If $n>1$, we cannot use the above described connection to a moduli space of colored graphs. This is due to the restrictions on edge-collapses which are not allowed to change the loop number of graphs. The resulting moduli spaces become thus cell complexes with ``missing faces'', also called \textit{faces at infinity}. As a consequence, the interpretation of $(HG,d)$ as the (simplicial) chain complex of a moduli space of graphs breaks down and we cannot use results on the topology of these spaces (which, for example, would guarantee the existence of non-trivial homology classes in certain degrees). 

However, since we are interested in cycles, not homology classes, we can construct cycles in $HG_{n,s}$ from cycles with lower loop numbers. There are two promising approaches: 
\begin{itemize}
 \item via the pre-Lie/operadic/dgla structure on Feynman graphs which by 
 \begin{equation*}
  d[G,H]=[dG,H]+(-1)^{|G|}[G,dH] 
 \end{equation*}
maps cycles to cycles; c.f.\ Thm.(\ref{dCCC}).
 \item via so-called assembly maps, as used in \cite{CoHaKaVo} in the context of Outer space, where new cycles are generated by gluing together graphs along their legs, i.e., by maps
 \begin{equation*}
  HG_{n_1,s_1} \otimes \cdots \otimes HG_{n_k,s_k} \longrightarrow HG_{n,s} \text{ with } n>n_1 + \ldots +n_k, \ s< s_1 + \ldots+ s_k.
 \end{equation*}
 See \cite{CoHaKaVo} for details. In the presence of colors and legs, already the simplest examples become very bulky. For an unphysical example, consider two bubble graphs $B_2$ merged along their legs, forming a cycle with $n=3$, $s=0$.
\end{itemize}
A detailed study of these ideas is left to future work.

\subsection{General colored graphs}

In principle we may set up a similar machine for the case where two or more colors/masses are equal. The only, but severe, complication is that this introduces symmetries via graph automorphisms into the picture. As a consequence, the corresponding graph complex detects too many relations because some graphs may cancel each other by symmetry reasons. 

\begin{eg}\label{eg:bananas}
 Let $B_k$ be the banana or melon graph on $k$-edges, all colored by the same color, 
 \begin{equation*}
  B_k=
 \raisebox{-.5cm}{ \begin{tikzpicture}[scale=0.7]
  \coordinate  (v1) at (0.3,0); 
   \coordinate  (v2) at (2.5,0);
   \coordinate (p1) at (0,0); 
   \coordinate (p2) at (2.8,0); 
   \draw (p1) to (v1) node[xshift=-.45cm]{$p_1$};
   \draw (p2) to (v2);
   \draw (v1) to[out=80,in=100] (v2) node[xshift=.5cm]{$p_2$};
   \draw (v1) to[out=30,in=150] (v2) node[xshift=-.75cm,yshift=0.1cm]{$\vdots$};
   \draw (v1) to[out=-30,in=-150] (v2);
   \draw (v1) to[out=-80,in=-100] (v2);
  \filldraw[fill=black] (v1) circle (0.07);
  \filldraw[fill=black] (v2) circle (0.07);
  \end{tikzpicture}}.
 \end{equation*}
 Then $dB_k=0$ if and only if $k$ is even.
\end{eg}

However, if we consider only homology classes of top degree and $s$ large enough, this problem does not show up. For instance, for one loop graphs we get similar results as in the previous section.

\begin{eg}\label{eg:twomasses}
Let us consider a theory with two particle masses, $a$ and $b$. Using \eqref{eq:dtriangle} in Ex.(\ref{eg:holotriangle}) where $x,y,z \in \{a,b\}$ we see that the element 
\begin{equation*}
 X= \raisebox{-.6cm}{ \begin{tikzpicture}[scale=.8]
   \coordinate  (v1) at (0,0.5); 
   \coordinate  (v2) at (1,0);
   \coordinate (v3) at (1,1);
   \coordinate (p1) at (-.23,0.5); 
   \coordinate (p2) at (1.2,-0.1);
   \coordinate (p3) at (1.2,1.1);
   \draw (p1) to (v1);
   \draw (p2) to (v2);
   \draw (p3) to (v3);
   \draw[red] (v1) -- (v2) node[pos=0.5, below]{$a$};
   \draw[blue] (v2) -- (v3) node[pos=0.5, right]{$b$};
   \draw[red] (v1) -- (v3) node[pos=0.5, above]{$a$};
  \filldraw[fill=black] (v1) circle (0.04) node[xshift=-.33cm]{$p_1$};
  \filldraw[fill=black] (v2) circle (0.04) node[xshift=.4cm,yshift=-.1cm]{$p_2$};
  \filldraw[fill=black] (v3) circle (0.04) node[xshift=.4cm,yshift=.1cm]{$p_3$};
  \end{tikzpicture}}
  +
    \raisebox{-.55cm}{ \begin{tikzpicture}[scale=.8]
   \coordinate  (v1) at (0,0.5); 
   \coordinate  (v2) at (1,0);
   \coordinate (v3) at (1,1);
   \coordinate (p1) at (-.23,0.5); 
   \coordinate (p2) at (1.2,-0.1);
   \coordinate (p3) at (1.2,1.1);
   \draw (p1) to (v1);
   \draw (p2) to (v2);
   \draw (p3) to (v3);
   \draw[red] (v1) -- (v2) node[pos=0.5, below]{$a$};
   \draw[red] (v2) -- (v3) node[pos=0.5, right]{$a$};
   \draw[blue] (v1) -- (v3) node[pos=0.5, above]{$b$};
  \filldraw[fill=black] (v1) circle (0.04) node[xshift=-.33cm]{$p_1$};
  \filldraw[fill=black] (v2) circle (0.04) node[xshift=.4cm,yshift=-.1cm]{$p_2$};
  \filldraw[fill=black] (v3) circle (0.04) node[xshift=.4cm,yshift=.1cm]{$p_3$};
  \end{tikzpicture}}
  +  \raisebox{-.55cm}{ \begin{tikzpicture}[scale=.8]
   \coordinate  (v1) at (0,0.5); 
   \coordinate  (v2) at (1,0);
   \coordinate (v3) at (1,1);
   \coordinate (p1) at (-.23,0.5); 
   \coordinate (p2) at (1.2,-0.1);
   \coordinate (p3) at (1.2,1.1);
   \draw (p1) to (v1);
   \draw (p2) to (v2);
   \draw (p3) to (v3);
   \draw[blue] (v1) -- (v2) node[pos=0.5, below]{$b$};
   \draw[red] (v2) -- (v3) node[pos=0.5, right]{$a$};
   \draw[red] (v1) -- (v3) node[pos=0.5, above]{$a$};
  \filldraw[fill=black] (v1) circle (0.04) node[xshift=-.33cm]{$p_1$};
  \filldraw[fill=black] (v2) circle (0.04) node[xshift=.4cm,yshift=-.1cm]{$p_2$};
  \filldraw[fill=black] (v3) circle (0.04) node[xshift=.4cm,yshift=.1cm]{$p_3$};
  \end{tikzpicture}}
  +
    \raisebox{-.55cm}{ \begin{tikzpicture}[scale=.8]
   \coordinate  (v1) at (0,0.5); 
   \coordinate  (v2) at (1,0);
   \coordinate (v3) at (1,1);
   \coordinate (p1) at (-.23,0.5); 
   \coordinate (p2) at (1.2,-0.1);
   \coordinate (p3) at (1.2,1.1);
   \draw (p1) to (v1);
   \draw (p2) to (v2);
   \draw (p3) to (v3);
   \draw[red] (v1) -- (v2) node[pos=0.5, below]{$a$};
   \draw[red] (v2) -- (v3) node[pos=0.5, right]{$a$};
   \draw[red] (v1) -- (v3) node[pos=0.5, above]{$a$};
  \filldraw[fill=black] (v1) circle (0.04) node[xshift=-.33cm]{$p_1$};
  \filldraw[fill=black] (v2) circle (0.04) node[xshift=.4cm,yshift=-.1cm]{$p_2$};
  \filldraw[fill=black] (v3) circle (0.04) node[xshift=.4cm,yshift=.1cm]{$p_3$};
  \end{tikzpicture}}
\end{equation*}
is $d$-closed. Inspecting Landau's equations for the first three graphs in the linear combination $X=G_1 + \ldots + G_4$ we find for the analytic function $\Phi(G_1 + G_2 + G_3)$ reduced singularities at $p_i^2=4a^2$, $p_i^2=0$ as well as $p_i^2=(a \pm b)^2$, $i=1,2,3$. The element $G_1+G_2+G_3$ is not $d$-closed, indicating that this sum is not ``complete'' with respect to this set of singularities. Indeed, we can add $G_4$ which has reduced singularities also along $p_i^2=4a^2$ and $p_i^2=0$. 
 \end{eg}

The preceding example holds in fact more generally. If we consider only one loop graphs with $s\geq 4$ legs and homology in degree greater than two, there are no automorphisms (no multi-edges and each vertex carries at least one leg-label). In this case we may mimic the constructions and arguments of the previous section.
In the general case, one has to keep an eye on possible symmetry-cancellations as in Ex.(\ref{eg:bananas}); see the discussion below.

We now introduce a variant of $(HG,d)$ that allows for general edge-colorings by elements of $\set m$ for $m \in \mb N$.

\begin{defn}\label{defn:cg}
 For $m,n,s \in \mb N$ define a chain complex $(CG,d)=(CG^m_{n,s},d)$ of $m$-colored graphs by
\begin{equation*}
CG=CG^m_{n,s}:= \mb Z_2 \big\langle (G,c) \mid G \in \mb G_{n,s},c:E_G \to \set{m} \big\rangle,
\end{equation*}
graded by $|(G,c)|:=|E_G|-1$, and equipped with the same differential $d$ as in Def.(\ref{defn:hg}),
\begin{equation*}
d(G,c):= \sum_{ e \in  E_G } (G/e,c_e).
\end{equation*}
\end{defn}

The basic results of the previous section, Lem.(\ref{lem:dsquared}) and Thm.(\ref{prop:cycles}), as well as all the points made thereafter, apply verbatim to the complexes $CG$. 

Moreover, for $n=1$ we have a similar interpretation of $(CG,d)$ as in the holocolored case. If $s\geq 4$ and we restrict attention to degree at least three, then this complex computes the corresponding homology groups of $\m{MCG}^m_{1,s}$, the \textit{moduli space of $m$-colored one loop graphs with $s$ legs}. For a detailed account of these moduli spaces we refer to \cite{MaxMarko}. In this case the results of \cite{MaxMarko} on the homology of $\m{MCG}^m_{1,s}$ (in degree greater or equal to three) may be used to find linear combinations of Feynman integrals that are $d$-closed, thus satisfy the property given in Thm.(\ref{prop:cycles}). 

Note that in regard to the connection to moduli spaces of graphs (or tropical curves) we retain for $m=1$ the classical (uncolored) cases of the latter spaces which are well studied in the mathematical literature \cite{ChanGalatiusPayne,V,rational}.

\begin{eg}
The computation in Ex.\ (\ref{eg:twomasses}) shows the existence of non-zero classes in $H_2(CG_{1,3}^m)$ for every $m\geq 2$. Furthermore, it implies that
\begin{equation*}
H_2(CG_{1,3}^m) \geq \mb Z_2^{m(m-1)}.
\end{equation*}
For $m=2$ this is an equality, $H_2(CG_{1,3}^2) \cong \mb Z_2^2$. For $m\geq 3$ it is a strict inclusion, because classes of the form constructed in Ex.\ (\ref{eg:holotriangle}) appear as well.
\end{eg}

For $m>1$ there exist only partial results on the homology of the moduli spaces of $m$-colored graphs $\m {MCG}_{1,s}^m$. Tab.(\ref{t:homdimc}) lists the known homology groups with rational coefficients, calculated with computer assistance (a list of generators can be found in \cite{MaxMaster} -- recall, that only for $s\geq 4$ and in degree greater than two this relates to the homology of the above defined complexes $CG$ (with $\mb Q$ replaced by $\mb Z_2$)).

\begin{table}[h!]
\centering
\begin{tabular}{ c || c | c | c | c | c  }
   & $H_0$ & $H_1$ & $H_2$ & $H_3$ & $H_4$ \\
	\hline  \\[-0.36cm]
	$\m {MCG}_{1,1}^2$ & 2 & - & - & - & - \\
$\m {MCG}_{1,2}^2$ & 1 & 0 & - & - & - \\
$\m {MCG}_{1,3}^2$ & 1 & 0 & 6 & - & - \\
$\m {MCG}_{1,4}^2$ & 1 & 0 & 3 & 9 & - \\
$\m {MCG}_{1,5}^2$ & 1 & 0 & 6 & 0 & 84 \\
\end{tabular}
\par\bigskip
\begin{tabular}{ c || c | c | c | c  }
   & $H_0$ & $H_1$ & $H_2$ & $H_3$ \\
	\hline \\[-0.36cm]
$\m {MCG}_{1,1}^3$ & 3 & - & - & - \\
$\m {MCG}_{1,2}^3$ & 1 & 1 & - & - \\
$\m {MCG}_{1,3}^3$ & 1 & 0 & 20 & - \\
$\m {MCG}_{1,4}^3$ & 1 & 0 & 3 & 103 \\
\end{tabular}
\quad
\begin{tabular}{ c || c | c | c | c  }
   & $H_0$ & $H_1$ & $H_2$ & $H_3$\\
	\hline \\[-0.36cm]
$\m {MCG}_{1,1}^4$ & 4 & - & - & - \\
$\m {MCG}_{1,2}^4$ & 1 & 3 & - & - \\
$\m {MCG}_{1,3}^4$ & 1 & 0 & 49 & - \\
$\m {MCG}_{1,4}^4$ & 1 & 0 & 3 & 426 \\
\end{tabular}
\par\bigskip
\begin{tabular}{ c || c | c | c   }
   & $H_0$ & $H_1$ & $H_2$  \\
	\hline \\[-0.36cm]
$\m {MCG}_{1,1}^5$ & 5 & - & -  \\
$\m {MCG}_{1,2}^5$ & 1 & 6 & -  \\
$\m {MCG}_{1,3}^5$ & 1 & 0 & 99  
\end{tabular}
\
\begin{tabular}{ c || c | c | c  }
   & $H_0$ & $H_1$ & $H_2$  \\
	\hline \\[-0.36cm]
$\m {MCG}_{1,1}^6$ & 6 & - & - \\
$\m {MCG}_{1,2}^6$ & 1 & 10 & - \\
$\m {MCG}_{1,3}^6$ & 1 & 0 & 176 
\end{tabular}
\
\begin{tabular}{ c || c | c | c }
   & $H_0$ & $H_1$ & $H_2$\\
	\hline \\[-0.36cm]
$\m {MCG}_{1,1}^7$ & 7 & - & - \\
$\m {MCG}_{1,2}^7$ & 1 & 15 & - \\
$\m {MCG}_{1,3}^7$ & 1 & 0 & 286
\end{tabular}
\vspace{0.5cm}
\caption{The dimensions of the homology groups $H_k(\m {MCG}_{1,s}^m;\mathbb{Q})$ for up to seven colors and various numbers of legs $s$.}
\label{t:homdimc}
\end{table}

Two interesting observations from \cite{MaxMarko}:
\begin{itemize}
 \item The top degree Betti numbers of $\m {MCG}_{1,s}^m$ (and hence also the number of classes in $H_{s-1}({CG_{1,s}^m})$) grow polynomially of degree $s$ as functions of the number of colors $m$ (see Theorem 4.13 in \cite{MaxMarko}).
 \item Conjecturally, only the top degree homology of $\m {MCG}_{1,s}^m$, or equivalently $CG_{1,s}^m$ (if $s\geq4$), depends on the number of colors, all other homology groups are independent of $m$ (see Conjecture 4.4 in \cite{MaxMarko}).
 On the level of Feynman integrals, with our interpretation given here, this appears less surprising. Introducing additional masses, while keeping the number of loops and legs fixed, changes only the constants in the corresponding Feynman integrals. Thus, this only recolors known cycles, except in the highest nontrivial degree where it generates new patterns of mass distributions in a Feynman graph. These patterns may give new homology classes (their number growing polynomially with $m$), while all new classes in lower degree come from reduced graphs, hence are exact.
\end{itemize}

For higher loop numbers the machinery introduced here may still be applied to the study of Feynman integrals, albeit with some restrictions. We find families exhausting a set of common singularities by computing the homology of $CG$, then checking which classes have representatives free of (color-respecting) automorphisms. 
However, it is important to note that for $m>1$ the homology of the complex $(CG^m_{n,s},d)$ or the space $\m{MCG}^m_{n,s}$ (as well as their relationship) is unknown so far.

\begin{rem}\label{twocomplexes}
 The results discussed here and in Sec.(\ref{prelieandccc}) relate two different chain complexes to the analytic structure of Feynman integrals, a cubical chain complex of graphs and spanning forests, and a ``simplicial" graph complex. Heuristically speaking, our results show that the former encodes more information about the analytic structure of Feynman integrals than the latter.
One is thereby led to wonder whether this fact is also reflected on the topological or homological level. 

For one loop graphs it is easy to see that the cubical chain complex arises from a cubical subdivision of the moduli space of (holo- or $m$-)colored graphs, hence it is indeed a finer structure.

This connection does not hold for higher loops, though. Here the cubical chain complex comes from a subdivision of a \emph{subspace} of the moduli space of colored graphs, a deformation retract, called the \textit{spine} in the context of Outer space (the uncolored case).\footnote{If one interprets Feynman integrals as volume forms on moduli spaces of graphs as in \cite{Marko}, then the results of Sec.(\ref{partialfractions}) show how the operation of deformation retracting gets balanced out by a simultaneous change of volume forms: When passing to the deformation retract, each cell, indexed by a graph $G$ is replaced by a cube complex, indexed by pairs $(G,T)$, where $T$ runs over all spanning trees of $G$, which is generally of lower dimension. However, Thm.(\ref{phiGT}) shows that an appropriate change of the associated volume form assures the equivalence of both constructions, that is, both give the same amplitude. See \cite{marko-ltd}.} 
In contrast, as we have seen in the proof of Thm.(\ref{prop:cycles}), the graph complex introduced here computes certain relative homology groups of a larger space\footnote{This space is one of two natural choices for compactifying moduli spaces of graphs. It is obtained by adding all simplices at infinity. The other choice is more intricate, a type of Borel-Serre compactification which is specifically suited to renormalization. See \cite{Marko} for the details.} that \emph{contains} the moduli space of colored graphs as a subspace.
It is thus not clear if and how the cubical chain complex can be understood as a refinement of the graph complex. 

This seems to be another instance of the well-known fact that there is a considerable jump in complexity when passing from the one loop case to higher loop numbers. However, at least on the topological level, this appears to be the only threshold.
Remarkably, the same is true ``in'' Outer space: Understanding the homology of the moduli spaces of one and two loop graphs (with legs) allows to construct (potential) classes in $H_*(\mathrm{Out}(F_n))$ for arbitrary large $n \in \mb N$; see \cite{CoHaKaVo}. 
\end{rem}

\subsection{Partitioning the one loop Green's function}\label{ss:partitiongreensfctn}

Recall that for the special\footnote{Actually, it plays quite a general role for Yang-Mills theories as was shown in \cite{KSvS}.} case of a theory with cubic interaction the graphs contributing to the $s$-point function (1PI) are the maximal degree elements of $HG_{n,s}$ or $CG_{n,s}^m$ (all vertices three-valent).

If $n=1$, then the maximal degree is $s$, so these elements are represented by colored cyclic graphs on $s$ edges.

For the holocolored case we immediately deduce from Thm.(\ref{prop:cycles}) and Eq.\eqref{eq:tophom} that the top degree homology classes in $HG_{1,s}$ form a partition of the set of graphs making up the one loop Green's function.

For general colorings we find this also to be true for $s=3$ and $m=2$: One class in $H_2(CG_{1,3}^2)$ is generated by the element $X$ in Ex.(\ref{eg:twomasses}), another by the same element with $a$ and $b$ interchanged. The graphs in these classes make up all of the graphs in $\mb G_{1,3}$ with two colors. Moreover, a simple calculation confirms that there are no other classes, hence $H_2(CG_{1,3}^2)\cong \mb Z_2^2$. Thus, the two cycles describe a partition of the one loop Green's function, 
\begin{equation*}
  \m A_{1,s}= a_1 + a_2  
\end{equation*}
with $a_1$ and $a_2$ as well as their singularities related by a $\Sigma_2$-symmetry, exchanging the colors $a$ and $b$.
If $m> 2$, then we find a partition of the degree two part of $CG_{1,3}^m$ by taking all classes $X$ as above for $a,b \in \set m$, $a \neq b$, together with the generator of $H_2(HG_{1,3})$ from Ex.(\ref{eg:holotriangle}) with $x,y,z \in \set m$, $x\neq y\neq z$ .
Note, however, that it is not clear whether this exhausts all homology classes. 

The case $s>3$ needs further study -- a starting point would be to use the list of the generators from \cite{MaxMaster} -- as does the question whether this holds for higher loop numbers as well.

\end{document}